\documentclass[12pt]{article}
\usepackage{amsthm}
\usepackage{amsmath}
\usepackage{ bbold }
\usepackage[colorlinks,citecolor=blue,urlcolor=blue,filecolor=blue,backref=page]{hyperref}
\usepackage{graphicx}
\usepackage[english]{babel}
\usepackage[utf8]{inputenc}
\usepackage{amsmath}
\usepackage{amssymb}
\usepackage{graphicx}
\usepackage{hyperref}
\usepackage{graphicx}
\usepackage{comment}
\usepackage{mathrsfs}
\usepackage{amsthm}
\usepackage{bbm}
\usepackage{scalerel}
\usepackage{rotating}
\usepackage{lscape}
\usepackage{subfigure}
\usepackage[colorinlistoftodos]{todonotes}
\usepackage{lipsum}
\usepackage{comment}
\usepackage[round]{natbib}
\usepackage[colorlinks,citecolor=blue,urlcolor=blue,filecolor=blue,backref=page]{hyperref}
\usepackage{graphicx}
\newtheorem{theorem}{Theorem}

\newtheorem{claim}[theorem]{Claim}
\numberwithin{equation}{section}
\theoremstyle{plain}

\newcommand{\ignore}[1]{}

\newcommand{\sF}{\mathcal{F}}
\newcommand{\sG}{\mathcal{G}}

\newcommand{\IE}{\mathbb{E}}

\newcommand{\IP}{\mathbb{P}}

\newcommand{\blind}{0}

\begin{document}
	
	\def\spacingset#1{\renewcommand{\baselinestretch}%
		{#1}\small\normalsize} \spacingset{1}

	\if0\blind
	{
		\title{\bf {Time-varying auto-regressive models for count time-series}}
		\author{Arkaprava Roy, 
			Sayar Karmakar\\
			University of Florida}
		\maketitle
	} \fi
	
	\if1\blind
	{
		\bigskip
		\bigskip
		\bigskip
		\begin{center}
			{\LARGE\bf }
		\end{center}
		\medskip
	} \fi
	
	\bigskip
	
	\begin{abstract}
		Count-valued time series data are routinely collected in many application areas. We are particularly motivated to study the count time series of daily new cases, arising from COVID-19 spread. First, we propose a Bayesian framework to study time-varying semiparametric AR$(p)$ model for count and then extend it to propose a time-varying INGARCH model considering the rapid changes in the spread. We calculate posterior contraction rates of the proposed Bayesian methods with respect to average Hellinger metric. Our proposed structures of the models are amenable to Hamiltonian Monte Carlo (HMC) sampling for efficient computation. We substantiate our methods by simulations that show superiority compared to some of the close existing methods. Finally we analyze the daily time series data of newly confirmed cases to study its spread through different government interventions.
	\end{abstract}
	
	\noindent%
	
	\spacingset{1.45}
	
	
	\section{Introduction}
	
	Modeling count time series is important in many applications such as disease incidence, accident rates, integer financial datasets such as price movement, etc. This relatively new research stream was introduced in \cite{zeger1988regression} and interestingly he analyzed another outbreak namely the US 1970 Polio incidence rate. This stream was furthered by \cite{chan1995} where Poisson generalized linear models (GLM) with an autoregressive latent process in the mean are discussed. A wide range of dependence was explored in \cite{davis2003} for simple autoregressive (AR) structure and external covariates. On the other hand, a different stream explored integer-valued time series counts such as ARMA structures as in \citep{brandt2001linear, biswas2009discrete} or INGARCH structure as done in \cite{zhu2011negative, zhu2012zero, zhu2012modeling1,zhu2012modeling}. However, from a Bayesian perspective, the only work to the best of our knowledge is that of \cite{silveira2015bayesian} where the authors discussed an ARMA model for different count series parameters. However, their treatment of ignoring zero-valued data or putting the MA structure by demeaned Poisson random variable remains questionable. None of these works focused on the time-varying nature of the coefficients except for a brief mention in \cite{sayar2020}.
	
	

	Our goals are motivated by both the application and methodological development. To the best of our knowledge, ours is the first attempt to model possibly autoregressive count time series with time-varying coefficients which can be regarded as the time-varying analog of \cite{fokianos2009poisson}. We consider a linear link based GLM route instead of the traditional exponential link \citep{fokianos2011log}, since 
	linear link helps in better interpretability of the coefficient functions. 
	Linear link however requires more stringent shape restrictions on the functions. 
	We impose those by putting constraints on the B-spline coefficients while modeling those coefficient functions. 
	However, it is possible to extend all the computations of the current paper to an exponential link based GLM framework.
	The mean function stands for the overall spread and the autoregressive coefficients stand for the effect of different lags. We are particularly motivated to study the spread of COVID-19 in New York City (NYC) from 23rd January to 14th July using the daily count data of new cases. In terms of our motivating data application, we wish to identify which lags are significant in our model which can be directly linked to the period of time symptoms did not show up. We find that some higher-order lags like 6, 7, and 8 are also significant. These findings are in-line with several research articles discussing the incubation length for the novel coronavirus with a median of 6-7 days and 98\% below 11 days. For example, see \cite{incubation}. We also find that after the lockdown or stay-at-home orders it takes about 12-16 days to reach the peak and then the intercept coefficient function starts decreasing. This is also an interesting find which characterizes the fact that the number of infected but asymptomatic cases is large compared to the new cases reported. 
	Additional to the time-varying AR model proposal, we also offer an analysis via a time-varying Bayesian integer-valued generalized autoregressive conditional heteroscedasticity (TVBINGARCH) model that assumes an additional recursive term in the conditional expectation (cf. \eqref{TVBING}). This extension offers some more comprehensiveness in the modeling part as even BINGARCH with small orders can help us get rid of choosing an appropriate maximum lag value. Since for a Poisson model, the mean is the same as the variance, this can also be thought of as an extension of the GARCH model in the context of count data. First introduced by \cite{ferland2006integer}, these models were thoroughly analyzed in \cite{zhu2012zero,zhu2012modeling1,zhu2011negative, zhu2012modeling, ahmad2016poisson}. Our proposal for the time-varying TVBINGARCH model adapts to the non-stationarity theme and also can be viewed as a new contribution. Finally, we contrast the time-varying AR and the GARCH for both simulations and real-data applications under different metrics of evaluation.
	Our semiparametric time-varying model provides better estimates.

	Regression models with varying coefficient were introduced by 
	\cite{hastie1993varying}. They modeled the varying coefficients using cubic B-splines. Later, these models has been further explored in various directions \cite{gu1993smoothing,biller2001bayesian,fan2008statistical,franco2019unified,yue2014bayesian}. Spline bases have been routinely used to model the time-varying coefficients within non-linear time series models \citep{cai2000functional, huang2002varying, huang2004functional,amorim2008regression}. We also consider the B-spline series based priors to model the time-varying coefficient functions. We develop efficient computational algorithms for the proposed models. 

	Apart from developing a computationally tractable hierarchical model, we also establish posterior contraction rates of the proposed models. \cite{ghosal2007convergence} established posterior contraction for a general stationary Markov chain with much stricter conditions. 
	However, they relaxed some of those conditions in Theorem 8.29 of \cite{ghosal2017fundamentals}.
	To the best of our knowledge, the posterior contraction rate result of this paper is the first for the time-varying Markov model based on minimal assumptions under Poisson-link. 
	We also consider the strategy, used to relax the conditions in Theorem 8.29 of \cite{ghosal2007convergence}.
	Our posterior contraction rate is with respect to the average Hellinger metric. The primary theoretical hurdle is to construct exponentially consistent tests in a time-varying Markov setup. Our proposed test construction is inspired by \cite{jeong2019frequentist, ning2020bayesian}. We construct the test relying on the Neyman-Pearson lemma with respect to negative average log affinity distance and calculate contraction rates. Then we show that the same rate holds for the average Hellinger metric as well. We also discuss a pointwise inferential tool by drawing credible intervals. Such tools are important to keep an objective perspective in terms of the evolution of the time-varying coefficients without restricting it to some specific trend models. See \cite{sayar2020}  (\cite{karmakar2018asymptotic} for an earlier version) for a comprehensive discussion on time-varying models and their applications. 
	
	The rest of the paper is organized as follows. Section \ref{model} describes the proposed Bayesian models in detail. Section~\ref{comp} discusses an efficient computational scheme for the proposed method. We calculate posterior contraction rates in Section~\ref{theo}. The performance of our proposed method in capturing true coefficient functions are studied in Section~\ref{sim} and we show excellent performance over other existing methods. Section~\ref{application} deals with an application of the proposed method on COVID-19 spread for NYC. Then, we end with discussions and possible future directions in Section~\ref{discussion}. Section~\ref{sec:proof} contains detail theoretical proofs.

	\section{Modeling}
	\label{model}
	Given the rapidly evolving nature of the pandemic, the patterns and number of new affected cases were changing rapidly over different geographical regions.  The rapid change in the observed counts make all earlier time-constant analysis inappropriate and builds a path where we can explore methodological and inferential development in tracking down the trajectory of this spread. Thus, we propose two novel semiparametric time-varying autoregressive models for counts to study the spread and examine the effects of these interventions in the spread based on the time-varying coefficient functions. 
	We first consider the most general case where we model the data using a time-varying Bayesian integer-valued generalized autoregressive conditional heteroscedasticity (TVBINGARCH) model where the conditional mean depends on the past observations as well as past conditional means. However, the relatively simpler process consisting of a time-varying mean/intercept function along with time-varying autoregressive coefficient functions upto lag-$p$ is also important keeping in mind the scope of application to real data and its interpretation. For example, the particular lags in an AR$(p)$ model for the COVID-19 count data can crave an interesting phenomenon in the lag-dynamics of the spread. This might be lost if we model the same using a TVBINGARCH(1,1) model since typically for GARCH type models it is standard practice to only consider smaller orders.

	\subsection{Time-varying generalized autoregressive conditional heteroscedasticity model for counts}
	\label{sec:TVBINGARCH}
	
	\cite{ferland2006integer} proposed integer valued analogue of generalized autoregressive conditional heteroscedasticity model (GARCH) after observing that the variability in number of cases of campylobacterosis
	infections also changes with level. 
	We consider here a time-varying analog of such process.
	The conditional distribution for count-valued time-series $X_t$ given $\sF_{t-1}=\{X_i: i\leq (t-1)\}$ and $\sG_{t-1}=\{\lambda_i: i\leq (t-1)\}$ is,
	\begin{align}
	X_t|\sF_{t-1},\sG_{t-1}\sim& \mathrm{Poisson}(\lambda_t) \text{ where } \lambda_t=\mu(t/T) + \sum_{i=1}^p a_i(t/T) X_{t-i}+\sum_{j=1}^q b_j(t/T)\lambda_{t-j}.\label{TVBING}
	\end{align}
	We call our method time-varying Bayesian Integer valued Generalized Auto Regressive Conditional Heteroscedastic (TVBINGARCH) model. We impose following constraints on the parameter space similar to \cite{ferreira2017estimation},
	\begin{align}\label{eq:parconditionG}
	\mathcal{P}_1=\{\mu, a_i:0<\mu(x)<\infty, \sup_{x}\sum_{i,j}(a_{i}(x)+b_j(x))<1\}.
	\end{align}
	This constraint ensure a unique solution of the time-varying GARCH process as discussed in \cite{davis2009extreme, rohan2013nonparametric, ferreira2017estimation}. Now, we put priors on the functions $\mu(\cdot)$, $a_i(\cdot)$ and $b_j(\cdot)$ such that they are supported in $\mathcal{P}_1$. Using the B-spline bases, we put following hierarchical prior on the unknown functions,
	\begin{align}
	\mu(x) =&\sum_{j=1}^{K_1}\alpha_jB_j(x)\label{prior3}\\
	a_{i}(x)=&\sum_{j=1}^{K_2}\theta_{ij}M_{i}B_j(x), \quad 0\leq\theta_{ij}\leq 1,1\leq i\leq p,\\
	b_{k}(x)=&\sum_{j=1}^{K_3}\eta_{kj}M_{k+p}B_j(x), \quad 0\leq\eta_{kj}\leq 1,1\leq k\leq q,\\
	M_i=&\frac{\tau_i}{\sum_{k=0}^p{\tau_k}}, \quad i=1,\ldots,p+q,\\
	\theta_{ij}\sim& U(0,1)\textrm{ for }1\leq i\leq p, 1\leq j\leq K_2,\\
	\eta_{kj}\sim& U(0,1)\textrm{ for }1\leq k\leq q, 1\leq j\leq K_3,\\
	\lambda_0\sim&\textrm{Inverse-Gamma}(d_1,d_1)\label{prior4},
	\end{align}
	where $\lambda_0$ is the rate parameter for $X_0$. The specification for the density of $X_0$ is required for computation. Otherwise we need to assume $\lambda_0$ to be known which is not reasonable for a real data application. We primarily focus on the special case where $p=1,q=1$. Based on the constraints on the parameter space we consider following prior for $\alpha_j$'s and $\tau_i$'s,
	\begin{align}
	\alpha_j\sim\mathrm{TN}(0,c_1^2, 0, \infty),\quad \tau_i \sim U(0,1), \label{prioringar:choice1}
	\end{align}
	\noindent where TN stands for the truncated normal with mean 0, variance $c_1^2$ and truncated to $[0,\infty)$.
	In above construction, $\sum_{j=0}^PM_j=1$. Thus $\sum_{j=1}^{p+q}M_j< 1$ if $M_0>0$. As $\Pi(M_0>0)=1$, we have $\Pi(\sum_{j=1}^{p+q}M_j< 1)=1$. Since $0\leq\theta_{ij}\leq 1$, we have $\sup_{x}a_i(x)\leq M_i$, and $\sup_{x}b_j(x)\leq M_{p+j}$. Thus $\sup_x\sum_{i=1}^{p}a_{i}(x)+\sum_{j=1}^{q}b_{j}(x)\leq \sum_{i=1}^{p+q}M_i< 1$. We have $\sum_{j=1}^{p+q}M_j= 1$ if and only if $\tau_0=0$, which has zero prior probability. On the other hand, we also have $\mu(\cdot)\geq 0$ as we have $\alpha_j\geq 0$. Thus, the induced priors, described in~\eqref{prior3}$-$~\eqref{prior4} are well supported in $\mathcal{P}_1$.

	\subsection{Time-varying auto-regressive model for counts}
	Although our previous modeling framework is more general, one may only wish to study higher order lag dependence from the past observations.
	Thus we consider a simplified model in this subsection.
	The linear Poisson autoregressive model \citep{zeger1988regression, brandt2001linear} is popular in analyzing higher order lag-dependence in count valued time series. Due to the assumed non-stationary nature of the data, we propose a time-varying version of this model. The conditional distribution for count-valued time-series $X_t$ given $\sF_{t-1}=\{X_i: i\leq (t-1)\}$ is,
	\begin{align}
	X_t|\sF_{t-1}\sim& \mathrm{Poisson}(\lambda_t) \text{ where } \lambda_t=\mu(t/T) + \sum_{i=1}^p a_i(t/T) X_{t-i}.\label{TVBARC}
	\end{align}
	We call our method time-varying Bayesian Auto Regressive model for Counts (TVBARC). The rescaling of the time-varying parameters to the support [0,1] is usual for in-filled asymptotics. Due to the Poisson link in~\eqref{TVBARC}, both conditional mean and conditional variance depend on the past observations. The conditional expectation of $X_t$ in the above model \eqref{TVBARC} is $\IE(X_t|\sF_{t-1})=\mu(t/T) + \sum_{i=1}^p a_i(t/T) X_{t-i}$, which needs to be positive-valued. To ensure that, we impose the following constraints on parameter space for the time-varying parameters, 
	\begin{align}\label{eq:parcondition}
	\mathcal{P}_2=\{\mu, a_i:0<\mu(x)<\infty, \sup_{x}\sum_{k}a_{k}(x)<1\}.
	\end{align}
	Note that, the conditions imposed (\ref{eq:parcondition}) on the parameters are somewhat motivated by the stationarity conditions for the time-constant versions of these models. This is not uncommon in time-varying AR literature. See \cite{dahlhaussubbarao2006, subbarao2008, sayar2020} for example. Even though the condition on $\mu(\cdot)$ seem restrictive in the light of what we need for invertible time-constant AR(p) process with Gaussian error, it is not unusual when it is used to model variance parameters to ensure positivity; it was unanimously imposed for all the literature mentioned above. Additionally, the above references heavily depend on local stationarity: namely, for every rescaled time  $0<t<1$, they assume the existence of an $\tilde{X}_i$ process which is close to the observed process. One key advantage of our proposal is it is free of any such assumption. Our assumption of only the first moment is also very mild for theoretical exploration in Section~\ref{theo}. Moreover, except for a very general linear model discussed in \citep{sayar2020}, to the best of our knowledge, this is the very first analysis of the time-varying parameter for count time-series modeled by Poisson regression. Thus we choose to focus on the methodological development rather than proving the optimality of these conditions. When $p=0$, our proposed model reduces to routinely used nonparametric Poisson regression model as in \cite{shen2015adaptive}.
	
	To proceed with Bayesian computation, we put priors on the unknown functions $\mu(\cdot)$ and $a_i(\cdot)$'s such that they are supported in $\mathcal{P}_2$. The prior distributions on these functions are induced through basis expansions in B-splines. Suitable constraints on the coefficients are imposed to ensure the shape constraints as in $\mathcal{P}_2$. Detailed description of the priors are given below,
	\begin{align}
	\mu(x) =&\sum_{j=1}^{K_1}\alpha_jB_j(x)\label{prior1}\\
	a_{i}(x)=&\sum_{j=1}^{K_2}\theta_{ij}M_{i}B_j(x), \quad 0\leq\theta_{ij}\leq 1,\\
	M_i=&\frac{\tau_i}{\sum_{k=0}^p{\tau_k}}, \quad i=1,\ldots,p,\\
	\theta_{ij}\sim& U(0,1)\textrm{ for }1\leq i\leq p, 1\leq j\leq K_2\label{prior2}.
	\end{align}
	Here $B_{j}$'s are the B-spline basis functions. The parameters $\delta_{j}$'s are unbounded. Based on the constraints on the parameter space we consider following prior for $\alpha_j$'s and $\tau_i$'s,
	\begin{align}
	\alpha_j\sim\mathrm{TN}(0, c_1^2, 0,\infty),\quad \tau_i \sim U(0,1), \label{prioringar:choice1}
	\end{align}
	where TN stands for the truncated normal distribution with mean 0, variance $c_1^2$ and truncated in $[0,\infty)$.
	The priors induced by above construction are $\mathcal{P}_2$-supported. The verification is very straightforward and similar to the previous subsection.

	\subsection{Model properties}
	In this paper, we only consider TVBINGARCH(1,1) which is commonly used for the GARCH class of models. One drawback of TVBARC is proper selection of lag. To alleviate this, one may then consider the TVBINGARCH framework. 
	As in the stationary case, TVBINGARCH(1,1) can be viewed as TVBARC with infinite order.
	Then the higher values in $b_1(\cdot)$'s is an indication that there might be important higher lags in TVBARC.
	Besides, to infer about higher lag dependence TVBARC is more suitable than TVBINGARCH. 
	In our real data illustration, we find that the TVBARC model identifies three important higher order lags 6,7 and 8 in COVID-19 spread.
	Such inference is difficult to obtain from TVBINGARCH.
	If the CH coefficient $b_1(\cdot)$ is uniformly zero, TVBINGARCH(1,1) reduces to TVBARC(1).
	However, the computational steps for TVBARC(1) does not easily follow from TVBINGARCH(1,1).
	Furthermore, our theoretical result of TVBINGARCH requires a lower bound for the true CH coefficient which is standard for time varying GARCH class of models. 
	Thus the theoretical result of TVBARC does not easily follow from TVBINGARCH.
	
	Towards writing the likelihood, note that our proposed models are non-stationary since the coefficient functions $a_i(\cdot)$, $b_j(\cdot)$ are possibly not constant. 
	However, we still take a simple product of individual conditional likelihoods for $X_t$'s rather than first locally approximating it by a stationary process. 
	The latter approach is more prominent in the frequentist framework and this phenomenon is known by `locally stationary approximation'. This was introduced in some seminal papers by \cite{dahlhaus1997fitting,dahlhaus2000likelihood} and were later used in many time-varying literature. See \cite{dahlhaussubbarao2006,dahlhaus2012locally,truquet2019local} among many others. Towards the Bayesian approach of modelling such approximating phenomenon, interested readers can refer to  \cite{rosen2009local,rosen2012adaptspec}. However, the assumption of existence of such an approximating stationary process is somewhat stringent and is probably not required in Bayesian paradigm. For example, see \cite{deyoreo2017bayesian} where the likelihood is formed by taking product of individual conditional likelihoods for a non-stationary time-series.
	Other approaches can be found in \cite{hadj2020bayesian,yang2020bayesian} where the likelihoods for the proposed non-stationary processes were computed without any local stationary approximation. Moreover, note that such approximating stationary processes can be shown to exist under the general smoothness conditions as outlined in Theorem 1 in \cite{dahlhaussubbarao2006} (for tvARCH case) or Proposition 2.3 in \cite{rohan2013nonparametric}(for tvGARCH case). These are easily extendible to the Poisson setting and for more general Holder smooth coefficient functions with probably an amended approximation rate. So in a sense, our smoothness assumption and the parameter restriction as (\ref{eq:parconditionG}) or (\ref{eq:parcondition})implies existence of such stationary processes without us implicitly putting additional assumption.

	\section{Posterior computation}
	\label{comp}
	In this section, we discuss the Markov Chain Monte Carlo (MCMC) sampling method for posterior computation. Our proposed sampling is dependent on the gradient-based Hamiltonian Monte Carlo (HMC) sampling algorithm \citep{neal2011mcmc}. Hence, we show the gradient computations of the likelihood with respect to different parameters for TVBARC$(p)$ and TVBINGARCH$(p,q)$ in the following two subsections.
	
	We obtain the likelihood from the joint density of the data based on our Poisson error model. 
	Since the joint density can be written as product of conditionals, we can thus write the joint likelihood of the data as product of conditional densities. 
	Detail expressions for each case are separately presented below.
	The likelihoods of the two models are constructed differently, thus we present them separately.

	\subsection{TVBINGARCH structure}
	We only derive the computational steps for TVBINGARCH(1,1) which is the frequent choice among GARCH-type models.
	While fitting this model, we assume for any $t<0$ $X_t=0,\lambda_t=0$. 
	The expression for $\lambda_1$ also involves $\lambda_0$.
	Thus, we need to additionally estimate the parameter $\lambda_0$, the Poisson rate parameter for $X_0$.
	Here the likelihood for TVBINGARCH(1,1) is given by $P(X_{0})\prod_{t=1}^TP(X_{t}|\mathcal{F}_{t-1})$.
	We assume that the marginal distribution of $X_{0}$ is Poisson$(\lambda_0)$ and the prior for $\lambda_0$ is Inverse-Gamma$(d_1, d_2)$ as described in Section~\ref{sec:TVBINGARCH}.
	The complete likelihood $L_2$ of the propose Bayesian method of~\eqref{TVBING} is given by
	\begin{align*}
	L_2&\propto \exp\bigg(\sum_{t=1}^T \big[-\{\mu(t/T) +  a_1(t/T) X_{t-1}+b_1(t/T) \lambda_{t-1}\big\} + X_t\log \big\{\mu(t/T) \\
	&\quad+  a_1(t/T) X_{t-1}+ b_1(t/T) \lambda_{t-i}\}\big] - \sum_{j=1}^{K_1} \alpha_j^2/(2c_1^2)\\
	&\quad-(d_1+1)\log\lambda_0 -d_1/\lambda_0\bigg){\mathbf 1}_{0\leq\theta_{11},\eta_{ij}\leq 1,, 0\leq \tau_{i}\leq 1, \alpha_{j}\geq 0},
	\end{align*}
	We calculate the gradients of negative log-likelihood $(-\log L_2)$ with respect to the parameters $\beta$, $\theta$, $\eta$ and $\delta$. The gradients are given below,
	\begin{align*}
	&-\frac{d\log L_2}{\alpha_1}\\&\quad=\bigg(1-\sum_t  \frac{B_1(t/T)X_{t-j}}{(\mu(t/T)+a_{1}(t/T)X_{t-j})+b_{1}(t/T)\lambda_{t-1})}\bigg) +\alpha_j/(2c_1^2),\\
	&-\frac{d\log L_2}{\theta_{11}}=M_{i}\bigg(1-\sum_t \frac{B_{1}(t/T)X_{t-j}}{(\mu(t/T)+a_{j}(t/T)X_{t-j})+b_{k}(t/T)\lambda_{t-1})}\bigg),\\
	&-\frac{d\log L_2}{\eta_{kj}}=M_{p+k}\bigg(1-\sum_t \frac{B_{1}(t/T)\lambda_{t-j}}{(\mu(t/T)+a_{j}(t/T)X_{t-j})+b_{k}(t/T)\lambda_{t-1})}\bigg),\\
	&-\frac{d\log L_2}{\tau_j}=\sum_k (M_j{\mathbf 1}_{\{j=k\}}-M_jM_k)\times \nonumber\\&\Bigg[\quad\sum_{i\leq p}\theta_{ij}B_j(x)\bigg(1-\sum_t \frac{B_{j}(t/T)X_{t-j}}{(\mu(t/T)+a_{j}(t/T)X_{t-j})+b_{1}(t/T)\lambda_{t-1})}\bigg){\mathbf 1}_{\{j\leq p\}}+\nonumber\\
	&\quad\sum_{1\leq k\leq q}\eta_{kj}B_j(x)\bigg(1-\sum_t \frac{B_{j}(t/T)\lambda_t}{(\mu(t/T)+a_{j}(t/T)X_{t-1})+b_{1}(t/T)\lambda_{t-1})}\bigg){\mathbf 1}_{\{j > p\}}\Bigg].
	\end{align*}
	The derivative of the likelihood concerning $\lambda_0$ is calculated numerically by differentiating from the first principles. Hence, it is sampled using the HMC algorithm too.
	
	\subsection{TVBARC structure}
	Since we do not have any information of the process for $t<0$, our computation for TVBARC(p) is based on the likelihood $\prod_{t=p}^TP(X_{t}|\mathcal{F}_{t-1})$.
	This likelihood may thus be regarded as a quasi-likelihood as we are looking at the joint density of last $T-p+1$ time points given the first $p$ observations and it is similar to the likelihood from \cite{deyoreo2017bayesian}.
	This likelihood also shares some commonality with the objective functions used for computation in \cite{dahlhaussubbarao2006, fryzlewicz2008normalized}.
	The complete posterior likelihood $L_1$ of the proposed Bayesian method in~\eqref{TVBARC} is given by
	\begin{align*}
	L_1&\propto \exp\bigg(\sum_{t=p}^T \big[-\{\mu(t/T) + \sum_{i=1}^p a_i(t/T) X_{t-i}\big\} + X_t\log \big\{\mu(t/T) \\
	&\quad+ \sum_{i=1}^p a_i(t/T) X_{t-i}\}\big] - \sum_{j=1}^{K_1} \alpha_j^2/(2c_1^2)\bigg){\mathbf 1}_{0\leq\theta_{ij}\leq 1, 0\leq \tau_{i}\leq 1, \alpha_{j}\geq 0},
	\end{align*}
	where we have $\mu(x) =\sum_{j=1}^{K_1}\exp(\beta_j)B_j(x), a_{i}(x)=\sum_{j=1}^{K_2}\theta_{ij}M_{i}B_j(x)$ and $M_j=\frac{\exp(\delta_j)}{\sum_{k=0}^p{\exp(\delta_k)}}$. We develop efficient MCMC algorithm to sample the parameter $\beta,\theta$ and $\delta$ from the above likelihood. The derivatives of above likelihood with respect to the parameters are easily computable. This helps us to develop an efficient gradient-based MCMC algorithm to sample these parameters. We calculate the gradients of negative log-likelihood $(-\log L_1)$ with respect to the parameters $\beta$, $\theta$ and $\delta$. The gradients are given below,
	\begin{align*}
	&-\frac{d\log L_1}{\alpha_j}=\bigg(1-\sum_t  \frac{B_j(t/T)X_t}{(\mu(t/T)+\sum_{j}a_{j}(t/T)X_{t-j})}\bigg) +\alpha_j/(2c_1^2),\\
	&-\frac{d\log L_1}{\theta_{ij}}=M_{i}\bigg(1-\sum_t \frac{B_{j}(t/T)X_t}{(\mu(t/T)+\sum_{j}a_{j}(t/T)X_{t-j})}\bigg),\\
	&-\frac{d\log L_1}{\tau_j}=\quad\sum_k (M_j{\mathbf 1}_{\{j=k\}}-M_jM_k)\sum_i\theta_{ij}B_j(x)\bigg(1-\sum_t \frac{B_{j}(t/T)X_{t-j}}{(\mu(t/T)+\sum_{j}a_{j}(t/T)X_{t-j})}\bigg),
	\end{align*}
	where ${\mathbf 1}_{\{j=k\}}$ stands for the indicator function which takes the value one when $j=k$.

	As the parameter spaces of $\theta_{ij}$'s and $\eta_{kj}$'s have bounded support, we map any Metropolis candidate, falling outside of the parameter space back to the nearest boundary point of the parameter space. To obtain a good acceptance rate, we tune our HMC sampler periodically. There are two tuning parameters in HMC namely the leapfrog step, and the step size parameter. The step size parameter is tuned to maintain an acceptance rate within the range of 0.6 to 0.8. The step size is reduced if the acceptance rate is less than 0.6 and increased if the rate is more than 0.8. This adjustment is done automatically after every 100 iterations. However, we choose to pre-specify the leapfrog step at 30 and obtain good results. Due to the increasing complexity of the parameter space in TVBINGARCH, we consider updating all the parameters involved in $a_{i}(\cdot)$'s, $b_k(\cdot)$'s, and $\lambda_0$ together. 
	
	\section{Large-sample properties}
	\label{theo}
	In this section we obtain posterior contraction rates for the two proposed models.
	Posterior contraction measures the speed at which we can recover the true parameter from the posterior distribution with increasing sample size. 
	The notion of recovery is specified by a semimetric $d$.
	
	{\underline {Definition \citep{ghosal2017fundamentals}}:} The posterior contraction rate at the true parameter $\kappa_0\in\mathcal{A}$ with respect to the semimetric $d$ on $\mathcal{A}$ is a sequence $\epsilon_T\to 0$ such that $P_{\kappa_0} \Pi (\kappa: d(\kappa,\kappa_0)>M_T \epsilon_T | X^{(T)} )\to 0$ for every $M_T\to \infty$, where $\mathcal{A}$ denotes the parameter space of $\theta_0$. Here $X^{(T)}$ stands for the complete dataset. 
	
	Although TVBINGARCH(1,1) may reduce to TVBARC(1) assuming $b_{1}(x)=0$ for all $x\in[0,1]$, the required technical assumptions do not allow us to derive the results for TVBARC as a special case for TVBINGARCH.
	For clarity in presenting the assumptions under which the respective results are established, we will make the conditions in (\ref{eq:parcondition}) and (\ref{eq:parconditionG}) more specific.  
	Since TVBARC is a simpler model, we first develop the theoretical results for this model and then make modifications to obtain the results for TVBINGARCH.

	\subsection{TVBARC structure}
	We start by studying large sample properties of the simpler AR model in~\eqref{TVBARC}. For simplicity, we fix order $p$ at $p=1$ for this section however the results are easily generalizable for any fixed order $p$ with some additional assumptions. The posterior consistency is studied in the asymptotic regime of increasing sample size $T$.  Let $\kappa=(\mu, a_1)$ stands for the complete set of parameters.
	For sake of generality of the method, we put a prior on $K_1$ and $K_2$ with probability mass function given by, 
	\begin{eqnarray}\label{eq:bij}
	\Pi(K_i=k)=b_{i1}\exp[-b_{i2} k (\log k)^{b_{i3}}],
	\end{eqnarray}
	with $b_{i1},b_{i2}>0$ and $0\le b_{i3}\le 1$ for $i=1,2$. Poisson and geometric probability mass functions appear as special cases of the above prior density for $b_{i3}=1$ or $0$ respectively. These priors have not been considered while fitting the model as it would require computationally expensive reversible jump MCMC strategy. We study the posterior consistency with respect to the average Hellinger distance on the coefficient functions which is 
	$$d_{1,T}^2=\frac{1}{T}d^2_{H}(\kappa_1,\kappa_2)=\frac{1}{T}\int(\sqrt{f_1}-\sqrt{f_2})^2,$$
	where $f_1=\prod_{t=1}^TP_{\kappa_1}(X_{t}|X_{t-1})$.
	Here $P$ stand for the conditional Poisson density defined in~\eqref{TVBARC}. 
	The contraction rate will depend on the smoothness of true coefficient functions $\mu$ and $a$ and the parameters $b_{13}$ and $b_{23}$ from the prior distributions of $K_1$ and $K_2$. Let $\kappa_0=(\mu_0, a_{10})$ be the truth of $\kappa$.
	
	\noindent Assumptions (A): There exists constants $0<M_{\mu}<M_X$ such that,
	\begin{itemize}
		\item[(A.1)] At time $t=0$, $\IE_{\kappa_0}(X_0)<M_{X}$. 
		\item[(A.2)] The coefficient functions $\sup_{x\in[0,1]} \mu_0(x)<M_{\mu}$ and $\sup_{x\in[0,1]} a_{10}(x) <1-M_{\mu}/M_X$.
		\item[(A.3)] $\inf_{x\in[0,1]} \min(\mu_0(x),a_{10}(x))>\rho$ for some small $\rho>0$.
	\end{itemize}
	Assumptions (A.1), (A.2) ensure $$\IE_{\kappa_0}(X_t)=\IE_{\kappa_0}(\IE_{\kappa_0}(X_t|X_{t-1}))<M_{\mu}+\left(1-\frac{M_{\mu}}{M_X}\right)M_X<M_X$$
	by recursion. Assumption (A.3) is imposed to ensure strict positivity of parameters and is standard in time-varying literature that deals with such constrained parameters. 

	Posterior consistency theory studies recovery of the `true' parameter $\kappa_{0}$ with increasing sample size when the data is sampled from the distribution characterized by $\kappa_{0}$. Our notion of recovery is based on the average Hellinger metric $d_{1,T}^2$ defined above.
	
	\begin{theorem}
		\label{contheo}
		Under assumptions (A.1)-(A.3), let the true functions $\mu_0(\cdot)$ and $a_{10}(\cdot)$ be H\"older smooth functions with regularity level $\iota_1$ and $\iota_2$ respectively, then the posterior contraction rate with respect to the distance $d_{1,T}^2$ is 
		\begin{align*}
		\max\bigg\{&T^{-\iota_1/(2\iota_1+1)} (\log T)^{\iota_1/(2\iota_1+1)+(1-b_{13})/2},T^{-\iota_2/(2\iota_2+1)} (\log T)^{\iota_2/(2\iota_2+1)+(1-b_{23})/2}\bigg\}.
		\end{align*}
		
	\end{theorem}
	where $b_{ij}$ are specified in (\ref{eq:bij}).
	\noindent For the proof, the first step is to calculate posterior contraction rate with respect to average log-affinity $r_T^2(f_1,f_2)=-\frac{1}{T}\log\int f_1^{1/2}f_2^{1/2}$ and then show that $r_T^2(f_1,f_2)\lesssim\epsilon_T^2$ implies $\frac{1}{T}d_H^2(f_1,f_2) \lesssim \epsilon_T^2$. The average log-affinity provides a unique advantage to construct exponentially consistent tests leveraging on the famous Neyman-Pearson Lemma as has also been used in \cite{ning2020bayesian} for a multivariate linear regression setup under group sparsity. The proof is postponed to Section~\ref{sec:proof}. The proof is based on the general contraction rate result from \cite{ghosal2017fundamentals} and some results on B-splines based finite random series. 

	\subsection{TVBINGARCH structure}
	Next, we discuss the more comprehensive tvBINGARCH model \eqref{TVBING}. To maintain simplicity in the proof, we again assume $p=1,q=1$.
	Similar to the previous subsection, we put a prior on the number of Bspline bases, $K_i$ with probability mass function given by, $$\Pi(K_i=k)=b_{i1}\exp[-b_{i2} k (\log k)^{b_{i3}}],$$ with $b_{i1},b_{i2}>0$ and $0\le b_{i3}\le 1$ for $i=1,2,3$. Let us assume that $\psi=(\mu,a_1,b_1)$ be the complete set of parameters. We study the posterior consistency with respect to the Hellinger distance on the coefficient functions which is 
	$$d_{2,T}^2=\frac{1}{T}d^2_{H}(\psi_1,\psi_2)=\frac{1}{T}\int(\sqrt{f_1}-\sqrt{f_2})^2,$$
	where $f_1=P_{\phi_1}(X_{0})\prod_{t=1}^TP_{\psi_1}(X_{t}|X_{t-1},\lambda_{t-1})$. Here $P$ stands for the conditional Poisson density defined in~\eqref{TVBINGARCH} and the marginal density of $X_{0}$, $P_{\phi_1}(X_{0})$ is Poisson$(\lambda_{10})$ as described in our computational steps.
	

	\noindent For this structure, we modify the assumptions as 
	
	\noindent Assumptions(B): There exists constants $0<M_{\mu}<M_X$ such that,
	\begin{itemize}
		\item[(B.1)] At time $t=0$, $\IE_{\psi_0}(X_0),\lambda_0<M_{X}$. 
		\item[(B.2)] The coefficient functions $\sup_{x\in[0,1]} \mu_0(x)<M_{\mu}$ and $\sup_{x\in[0,1]} (a_{10}(x)+b_{10}(x)) <1-M_{\mu}/M_X$.
		\item[(B.3)] $\inf_{x\in[0,1]} \min(\mu_0(x),a_{10}(x), b_{10}(x))>\rho$ for some small $\rho>0$.
	\end{itemize}
	Assumptions (B.1), (B.2) ensure $$\IE_{\psi_0}(X_t)=\IE_{\psi_0}(\IE_{\psi_0}(X_t|X_{t-1},\lambda_{t-1}))<M_{\mu}+\left(1-\frac{M_{\mu}}{M_X}\right)M_X<M_X$$
	by recursion. 
	Thus we have, by Assumption (B.1-B.2)
	$$ \IE_{\psi_0}(X_t) <M_X,\quad  \IE_{\psi_0}(\lambda_t)=\IE_{\psi_0}(X_t|X_{t-1},\lambda_{t-1})=\IE_{\psi_0}(X_t) <M_X.$$

	\noindent Assumption (B.3) is imposed to ensure strict positivity of parameters and is standard in time-varying literature that deals with such constrained parameters. 
	Now we present our posterior contraction rate theorem below. 
	The definition of the contraction rate is the same as before.
	
	\begin{theorem}
		\label{contheo2}
		Under assumptions (B.1)-(B.3), let the true functions $\mu_0(\cdot)$, $a_{10}(\cdot)$ and $b_{10}(\cdot)$ be H\"older smooth functions with regularity level $\iota_1$, $\iota_2$ and $\iota_3$ respectively, then the posterior contraction rate with respect to the distance $d_{2,T}^2$ is
		\begin{align*}
		\max\bigg\{&T^{-\iota_1/(2\iota_1+1)} (\log T)^{\iota_1/(2\iota_1+1)+(1-b_{13})/2},T^{-\iota_2/(2\iota_2+1)} (\log T)^{\iota_2/(2\iota_2+1)+(1-b_{23})/2},\\&T^{-\iota_3/(2\iota_3+1)} (\log T)^{\iota_3/(2\iota_3+1)+(1-b_{33})/2}\bigg\}.
		\end{align*}
		
	\end{theorem}
	
	\noindent The proof follows from a similar strategy as in Theorem 1. An outline of the proof can be found in the Section~\ref{sec:proof}. 
	
	\section{Simulation studies}
	\label{sim}
	In this section, we study the performance of our proposed Bayesian method in capturing the true coefficient functions. We compare both TVBARC and TVBINGARCH methods with some other competing models. It is important to note that, this is to the best of our knowledge first work in Poisson autoregression with a time-varying link. Thus, we compare our method with the existing time-series models with time-constant coefficients for count data and time-varying AR with Gaussian error. We also examine the estimation accuracy of the coefficient functions for estimating the truth. 
	
	The hyperparameter $c_1$ of the truncated normal prior is set to 10 to ensure weak informativeness. The hyperparameters for Inverse-Gamma prior $d_1=0.1$, which is also weakly informative. We consider 6 equidistant knots for the B-splines based on comparing the AMSE scores. We choose the knot number after which the AMSE score does not change significantly. We collect 10000 MCMC samples and consider the last 5000 as post-burn-in samples for inferences. In absence of any alternative method for time-varying AR$(p)$ model of count-valued data, we compare the estimated functions with the true functions in terms of the posterior estimates of functions along with its 95\% pointwise credible bands. The credible bands are calculated from the MCMC samples at each point $t=1/T,2/T,\ldots, 1$. We also compare different competing methods in terms of average MSE (AMSE) score using the INGARCH method of {\tt tsglm} from R package {\tt tscount}, GARMA using {\tt tscount} as well, {\tt tvAR} and our proposed Bayesian methods. The AMSE is defined as $\frac{1}{T}\sum_t(X_t-\hat{\lambda}_t)^2$. We estimate this in terms of the posterior mean of AMSEs across MCMC as 
	$$AMSE=\frac{1}{5000}\sum_{S=1}^{5000}\frac{1}{T}\sum_t(X_t-\hat{\lambda}^S_{t})^2,$$ where $\hat{\lambda}^S_{t}$ is the posterior estimate of $\lambda_{t}$ at $S$-$th$ postburn sample.

	\subsection{Case 1: TVBARC structure}
	\label{TVARsec}
	
	Here, we consider two model settings $p=1;X_t\sim\textrm{Poisson}(\mu(t/T)+a_1(t/T)X_{t-1})$ and $p=2;X_t\sim\textrm{Poisson}(\mu(t/T)+a_1(t/T)X_{t-1} + a_2(t/T)X_{t-2})$ for $t=1,\ldots,T$. Three different choices for $T$ have been considered, $T=100,500$ and $1000$. The true functions are for $x\in[0,1]$,
	\begin{align*}
	\mu_0(x)=&10\exp\big(-(x-0.5)^2/0.1\big),\\
	a_{10}(x)=&0.3(x-1)^2+0.1,\\
	a_{02}(x)=&0.4x^2+0.1.
	\end{align*}
	
	We compare the estimated functions with the truth for sample size 1000 in Figures~\ref{AR1} and Figure~\ref{AR2} for the models $p=1$ and $p=2$ respectively. Tables~\ref{AMSEAR1} and~\ref{AMSEAR2} illustrate the performance of our method with respect to other competing methods.
	
	\begin{table}[ht]
		\centering
		\caption{AMSE comparison for different sample sizes across different methods when the true model is~\eqref{TVBARC} with $p=1$.}
		\begin{tabular}{rrrrr}
			\hline
			& INGARCH(1,0) & GARMA(1,0) & TVAR(1) & TVBARC(1) \\ 
			\hline
			$T=100$ & 11.60 & 11.18 & 11.41 & \textbf{8.65} \\  
			$T=500$ & 11.35 & 11.04 & 11.24 & \textbf{8.12} \\ 
			$T=1000$ & 11.05 & 10.73 & 10.94 & \textbf{7.02} \\ 
			\hline
		\end{tabular}
		\label{AMSEAR1}
	\end{table}
	
	\begin{table}[ht]
		\centering
		\caption{AMSE comparison for different sample sizes across different methods when the true model is~\eqref{TVBARC} with $p=2$.}
		\begin{tabular}{rrrrr}
			\hline
			& INGARCH(2,0) & GARMA(2,0) & TVAR(2) & TVBARC(2) \\ 
			\hline
			$T=100$ & 18.02 & 17.28 & 13.04 & \textbf{11.01} \\ 
			$T=500$ & 16.42 & 15.86 & 12.61 & \textbf{10.79} \\ 
			$T=1000$ & 15.79 & 15.25 & 12.75 & \textbf{10.61} \\ 
			\hline
		\end{tabular}
		\label{AMSEAR2}
	\end{table}
	
	\begin{figure}[htbp]
		\centering
		\subfigure[$\mu()$]{\label{fig:c.1}\includegraphics[width=50mm]{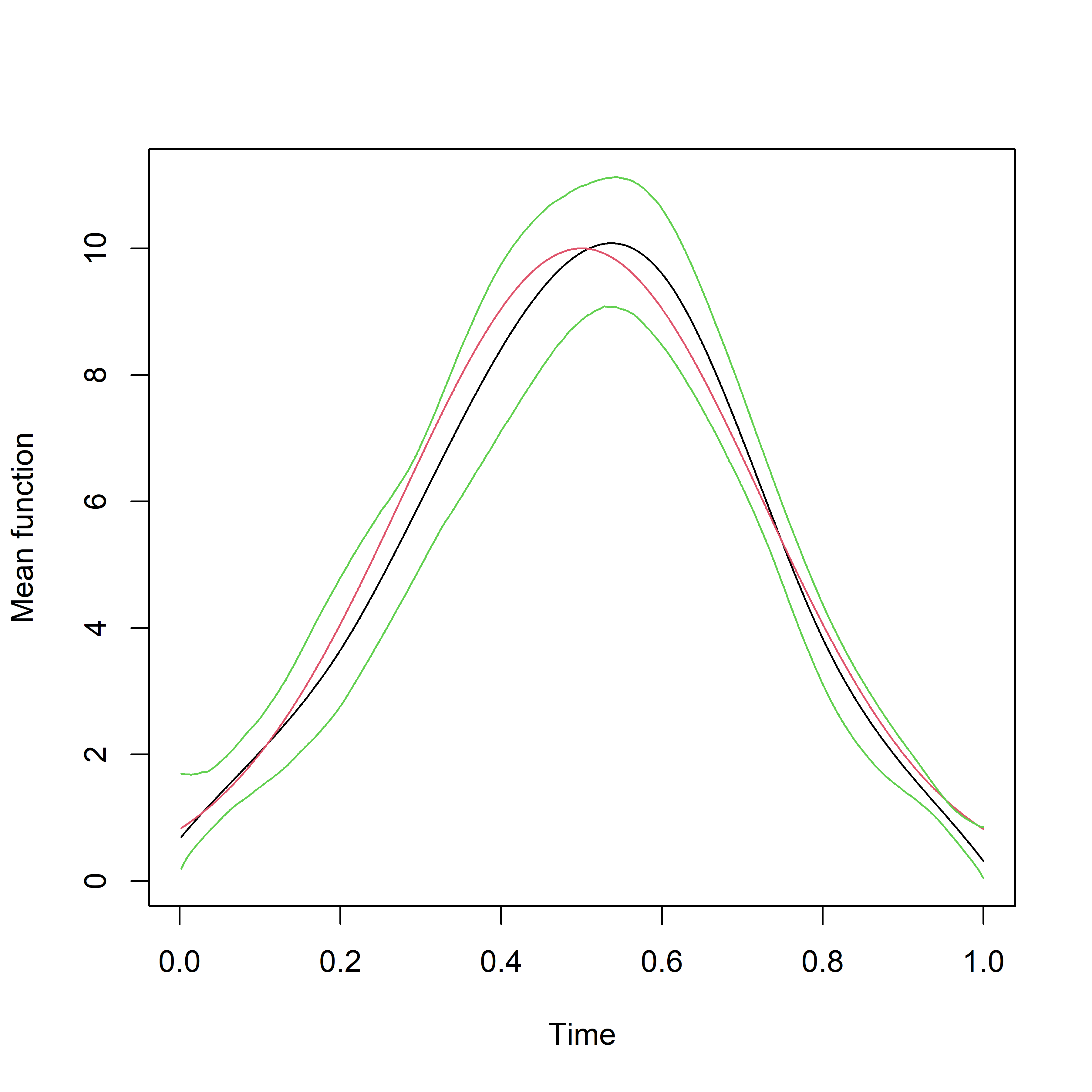}}
		\subfigure[$a_1()$]{\label{fig:c.2}\includegraphics[width=50mm]{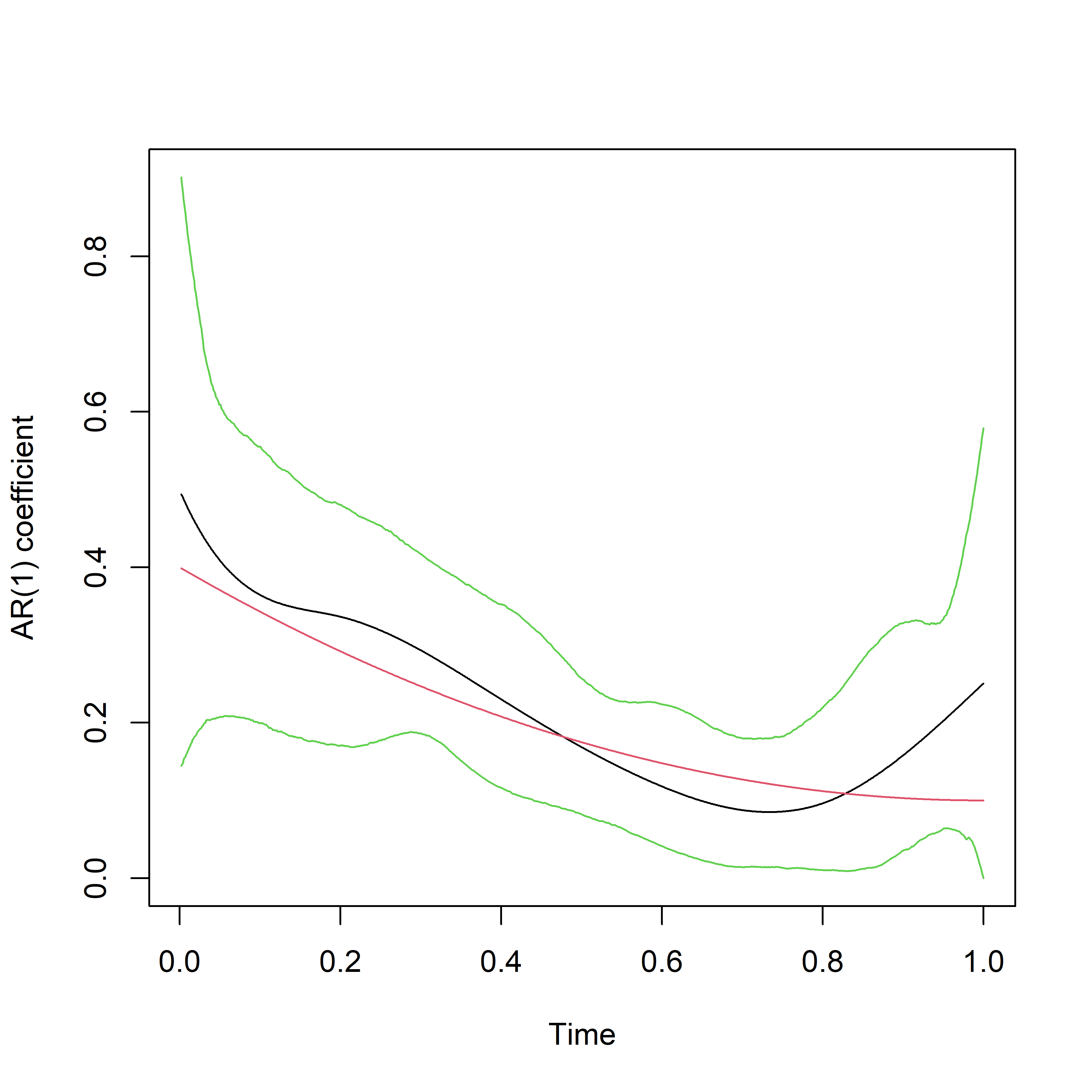}}
		\caption{Estimated mean function in 1st column and estimated AR(1) coefficient function in the 2nd column for the case $p=1$ and sample size 1000. Red is the true function, black is the estimated curve along with the 95\% pointwise credible bands in green.} 
		\label{AR1}
	\end{figure}

	\begin{figure}[htbp]
		\centering
		\includegraphics[width=120mm]{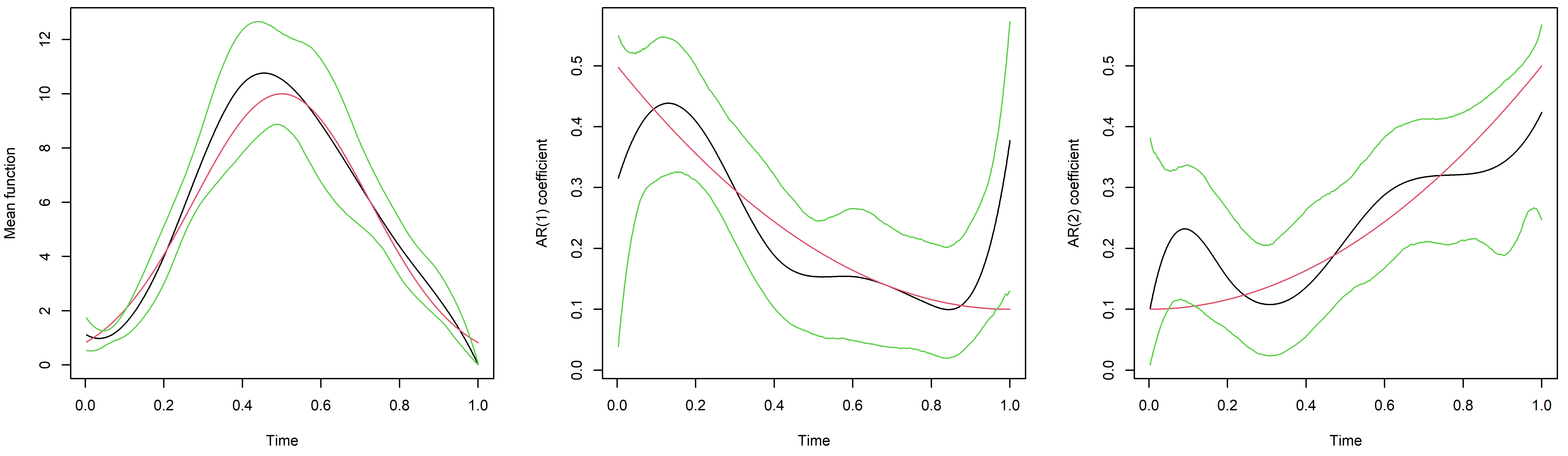}
		\caption{Estimated coefficient functions for the simulation case $p=2$ and sample size 1000. Red is the true function, black is the estimated curve along with the 95\% pointwise credible bands in green.} 
		\label{AR2}
	\end{figure}

	\subsection{Case 2: TVBINGARCH structure}
	\label{TVBINGARCHsec}
	For the tvBINGARCH case, we only consider one simulation settings $p=1, q=1; X_t\sim\textrm{Poisson}(\mu(t/T)+a_1(t/T)X_{t-1} +b_1(t/T)\lambda_{t-1})$. Two different choices for $T$ have been considered, $T=100$ and $200$ and  for $x\in[0,1]$ the coefficient functions are,
	\begin{align*}
	\mu_0(x)=&25\exp\big(-(x-0.5)^2/0.1\big),\\
	a_1(x) =& 0.3(x-1)^2+0.1,\\
	b_1(x) =& 0.1x^{1.5}+0.1
	\end{align*}
	
	\noindent Figure~\ref{TVBINGARCH} compares the estimated functions with the truth for sample size 200 for the model in~\eqref{TVBING} with $p=1,q=1$. The performance of our method is compared to other competing methods in Tables~\ref{TVINGARtab}.
	
	\begin{figure}
		\centering
		\includegraphics[width=100mm]{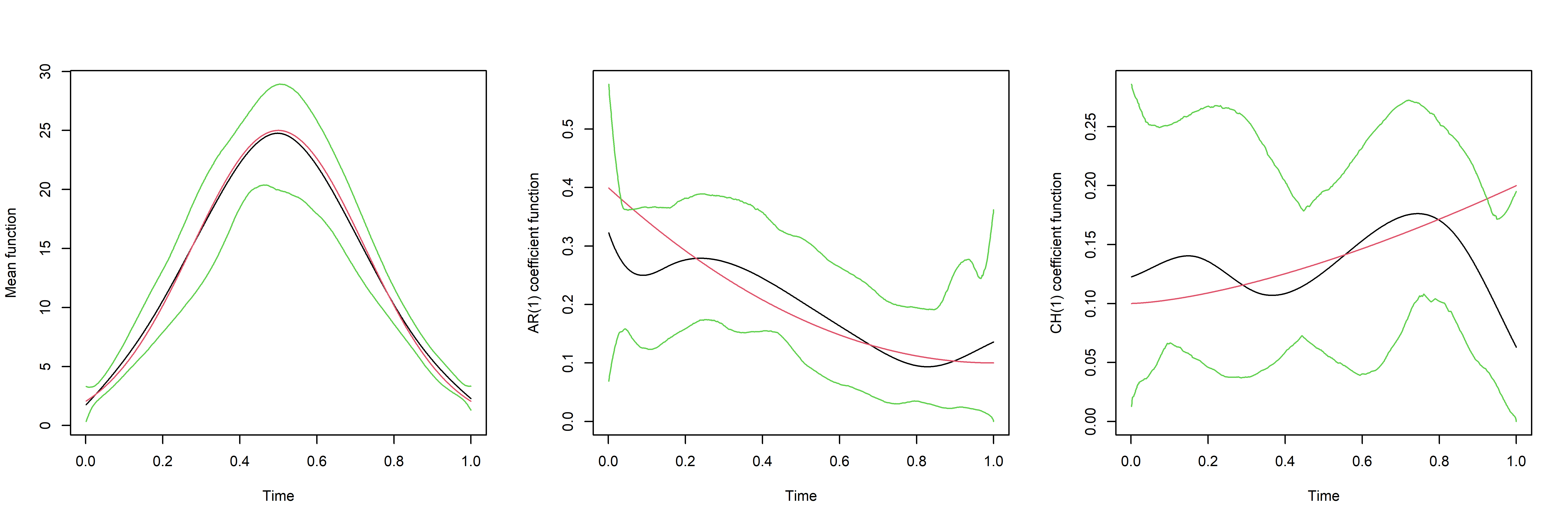}
		\caption{Estimated coefficient functions for the TVBINGARCH(1,1) and sample size 1000. Red is the true function, black is the estimated curve along with the 95\% pointwise credible bands in green.}
		\label{TVBINGARCH}
	\end{figure}
	
	\begin{table}[ht]
		\caption{Average MSE comparison for different sample sizes across different methods when the true model is~\eqref{TVBING} with $p=1, q=1$.}
		\centering
		\begin{tabular}{rrrrrr}
			\hline
			& INGARCH(1,1) & GARMA(1,1) & tvAR(10) & TVBINGARCH(1,1) \\ 
			\hline
			$T=100$ & 27.38 & 27.60 & 24.50 & \textbf{22.83} \\ 
			$T=500$ & 24.02 & 24.07 & 22.90 & \textbf{21.23} \\
			$T=1000$& 23.23 & 23.32 & 22.93 & \textbf{21.19}\\
			\hline
		\end{tabular}
		\label{TVINGARtab}
	\end{table}
	
	Figure~\ref{AR1} to~\ref{TVBINGARCH} shows that our proposed Bayesian method captures the true functions quite well for both of the two simulation experiments. We find that the estimation accuracy improves as the sample size increases. As the sample size grows, the 95\% credible bands are also getting tighter, implying lower uncertainty in estimation. This gives empirical evidence in favor of the estimation consistency which has also been verified theoretically in Section~\ref{theo}. 
	The average mean square error (AMSE) is always the lowest for our method in Tables~\ref{AMSEAR1}, \ref{AMSEAR2} and~\ref{TVINGARtab}. 
	
	\section{COVID-19 spread at NYC}
	\label{application}
	We collect the data of new affected cases for every day from 23rd January to 14th July from an open-source platform \url{ {https://www.kaggle.com/sudalairajkumar/novel-corona-virus-2019-dataset}}. The end date 14th July is chosen as around that time NYC started the process of re-opening. The data on daily new cases are illustrated in Figure~\ref{fig:NYCdata}. We were particularly interested in NYC data as this city remained an epicenter in US for about a month. With the help of government interventions and sustained lock-down, the recovery was significant in about 3 months. Such a time-varying nature of the data motivated us to retrospect as how the mean trend and AR trend behave which can also shed some insight about effects of lockdown or the contagious spread. 
	
	Based on the findings on the incubation of the virus in \cite{lauer2020incubation} and others, it is understood that the symptoms often take some time after the virus affects through contagion. Our idea is to consider different models with varying number of lags for this. We consider TVBARC(1), TVBARC(10) and TVBINGARCH(1,1) here. The results for the TVBARC(1) are illustrated in Figure~\ref{fig:NYC1}. We see that during the spike in daily new cases the function $a_1(\cdot)$ is the highest. Figure~\ref{fig:NYC10} depicts the estimated mean and coefficient functions from a TVBARC(10) model. We find that the estimated $a_1(\cdot)$ functions show a similar trend. On top of that, we see that $a_6(\cdot),a_7(\cdot)$ and $a_8(\cdot)$ have also some effect. Finally we fit our TVBINGARCH(1,1) which might be considered TVBARC with infinite order. Figure~\ref{fig:NYCTVBINGARCH} depicts the estimated functions, the mean $\{\mu(\cdot)\}$, AR(1) $\{a_1(\cdot)\}$ and CH(1) $\{b_1(\cdot)\}$ coefficient functions. In Table~\ref{tab:NYCcompare}, we compare the AMSE scores across different models. For all the models, we consider 12 equidistant knots based on the AMSE scores as discussed in Section~\ref{sim}. 
	
	Figure~\ref{fig:NYC10} suggests that even lag 6, 7, and 8 have some significant contribution. The effect of this lag is suppressed in Figure~\ref{fig:NYC1} and is expressed in terms of $b_1(\cdot)$ of Figure~\ref{fig:NYCTVBINGARCH}. The estimated mean functions also behave similarly for all the three cases. It shows a spike during the rise of daily new cases. After that, it decreases which can talk about successful containment strategies in NYC. More specifically it decreases after around 15 days since the strict implementation of statewide lockdown on 20-$th$ March. This is consistent with what was found in our unsubmitted preprint \citep{roy2020bayesian} through an empirical early-stage analysis of the spread in different cities and countries.
	
	The effect of Lag 6, 7, and 8 can be attributed to the incubation period of the virus. It can also lead to the finding that there was a weekly periodicity which is probably due to shorter testing/administrative facilities being available during the weekend. Note that our choice of fitting an TVBARC(10) model is more general than separately fitting a seasonal/periodic time-series model. Another important finding is coming from the overall trend of $a_1(\cdot)$. It starts to decrease when the number of cases starts going down. However later on it varies around 0.6 can be attributed to the fact that the number of new cases did not vary much and remained around the same level from the middle of May. The credible bands look very small around the mean function which is probably due to the large magnitude of the estimated function. 
	
	\begin{figure}
		\centering
		\includegraphics[width=100mm]{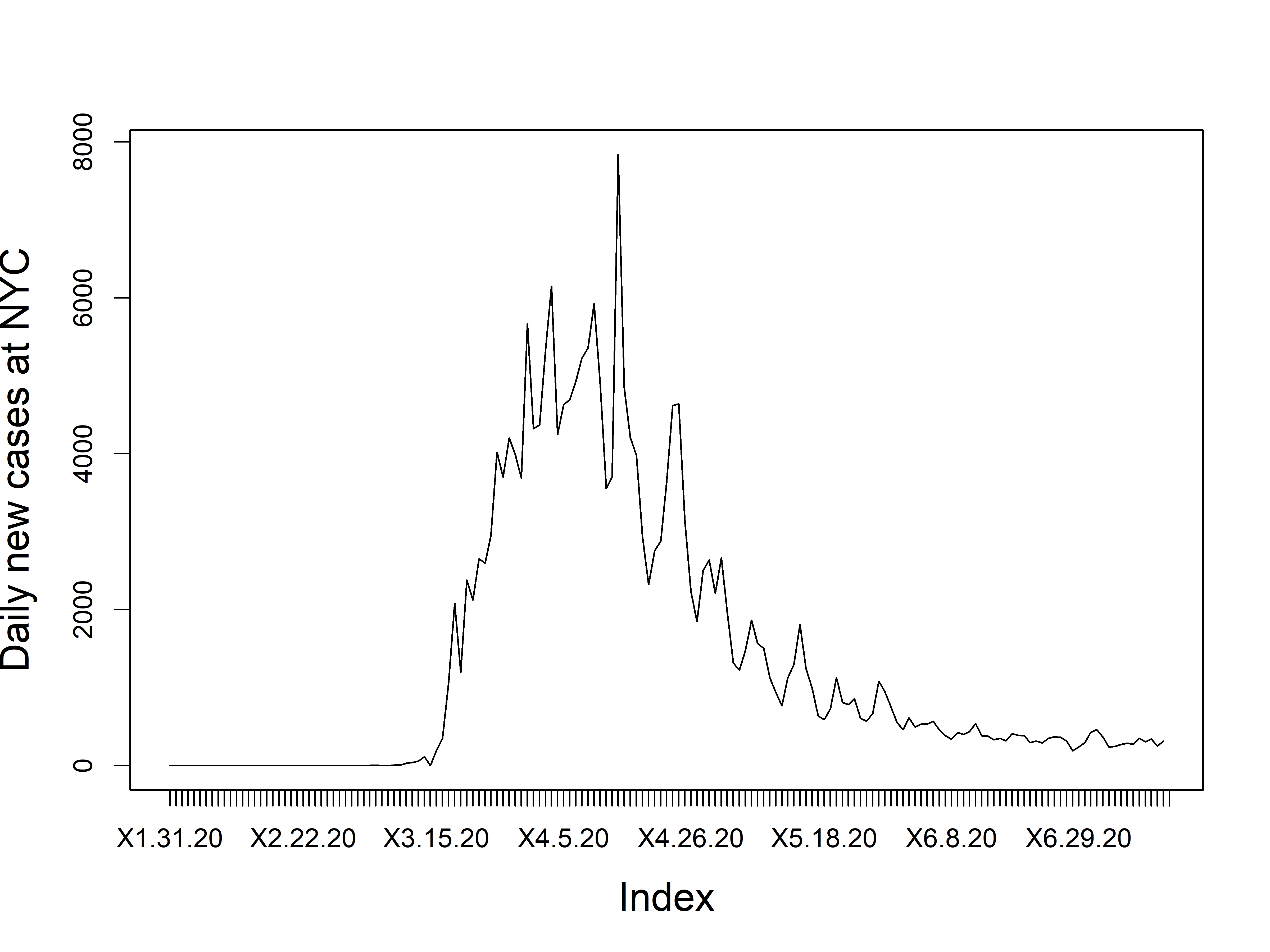}
		\caption{Daily new COVID-19 cases from 31st January to 14th of July recorded at NYC.}
		\label{fig:NYCdata}
	\end{figure}
	
	\begin{figure}[htbp]
		\centering
		\subfigure[NYC-$\mu(\cdot)$ function]{\label{fig:a.1}\includegraphics[width=60mm]{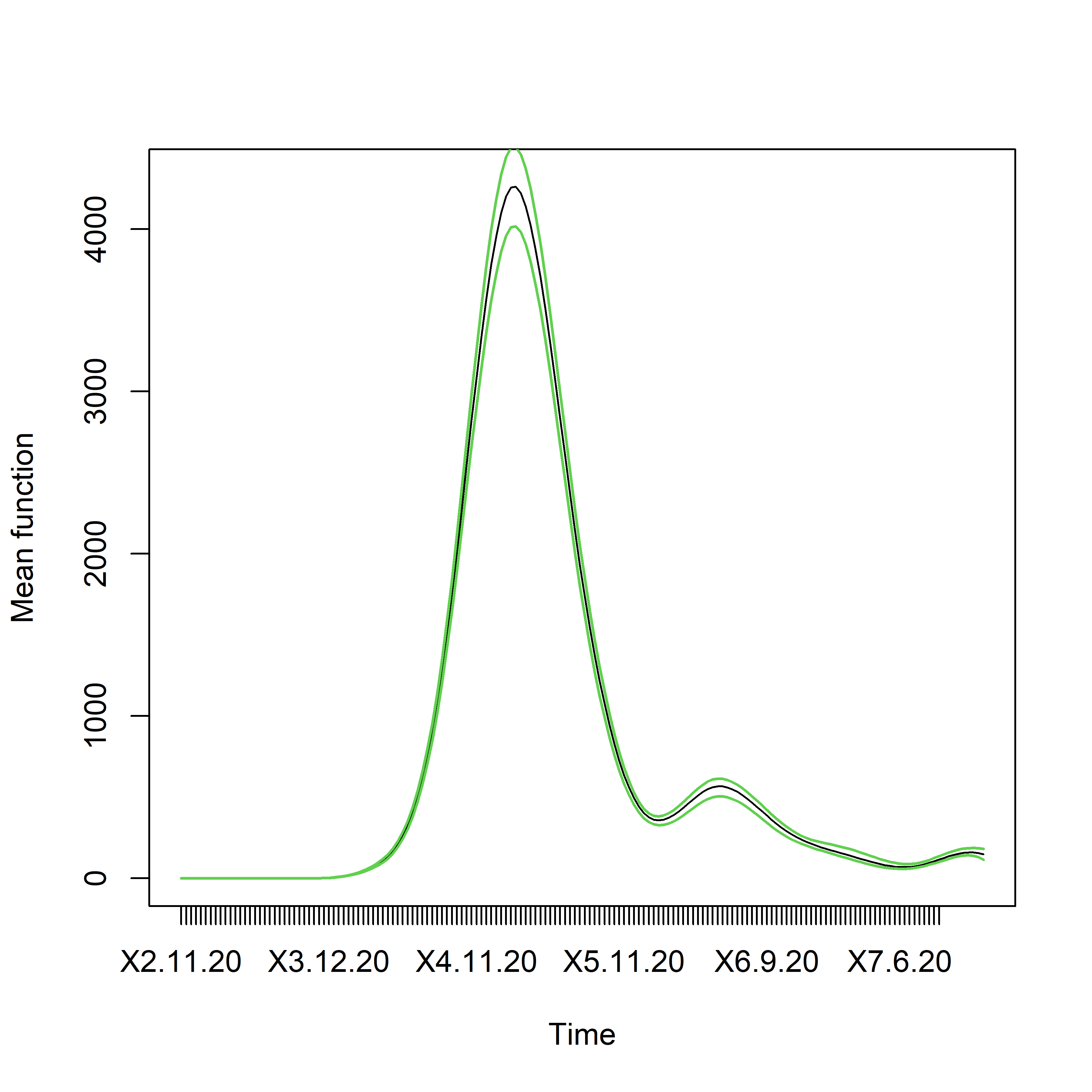}}
		\subfigure[NYC-$a_1(\cdot)$ function]{\label{fig:a.2}\includegraphics[width=60mm]{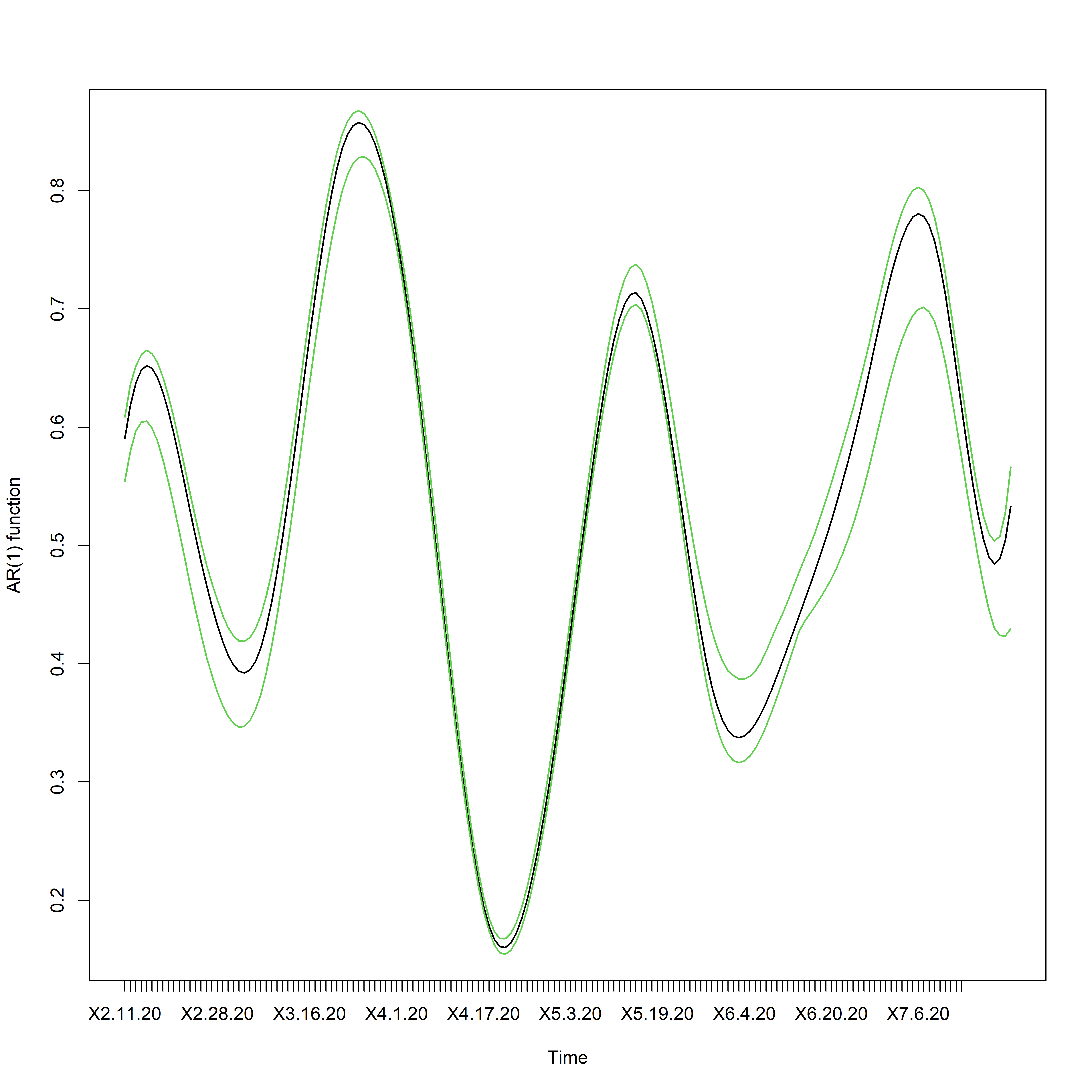}}
		\caption{Estimated mean functions in 1st column and estimated AR coefficient functions in the 2nd column for NYC using TVBARC(1). Black is the estimated curve along with the 95\% pointwise credible bands in green for the mean and AR(1) function.} 
		\label{fig:NYC1}
	\end{figure}
	
	\begin{figure}[htbp]
		\centering
		\subfigure[NYC-$\mu(\cdot)$ function]{\label{fig:a.1}\includegraphics[width=60mm]{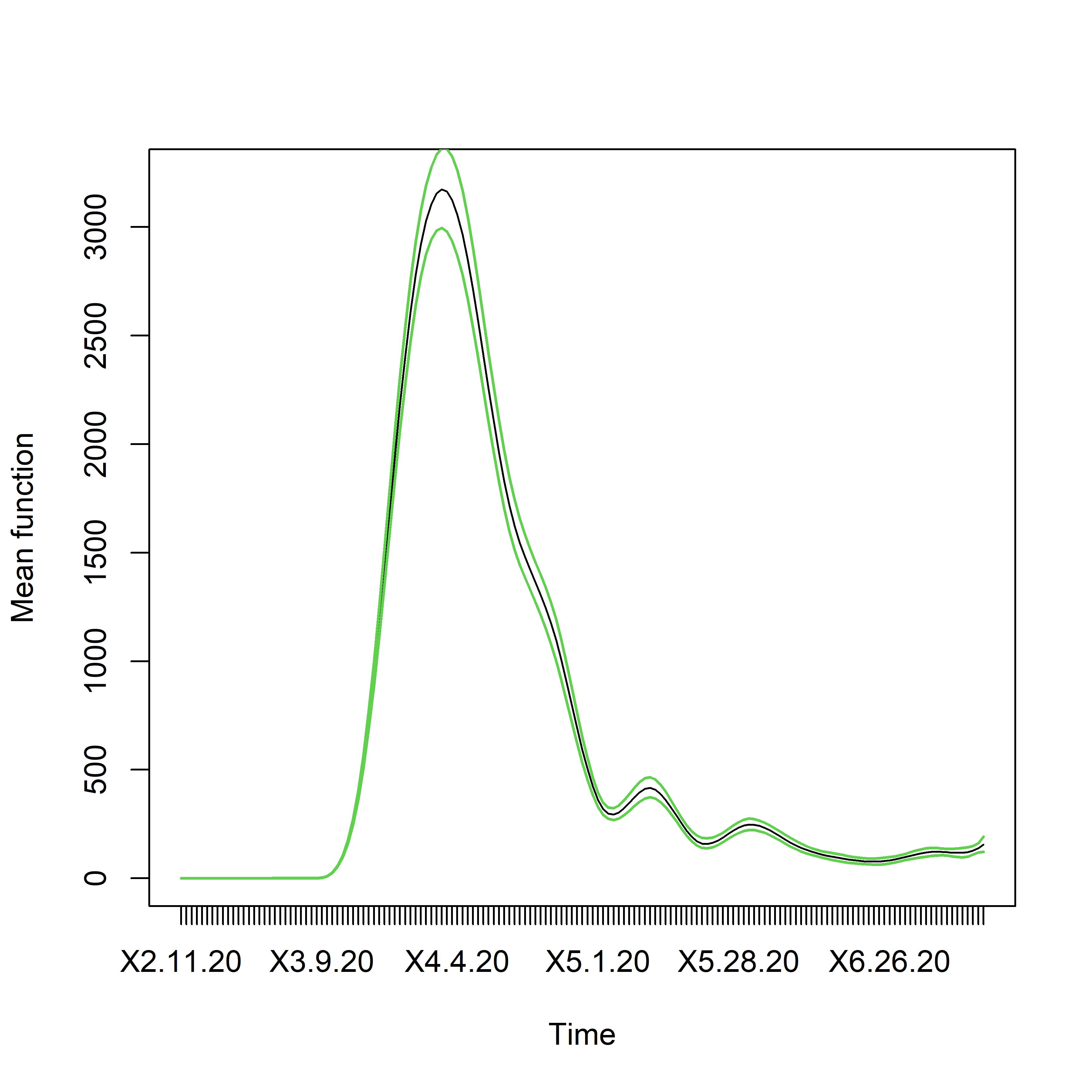}}
		\subfigure[NYC-$a(\cdot)$ functions]{\label{fig:a.2}\includegraphics[width=60mm]{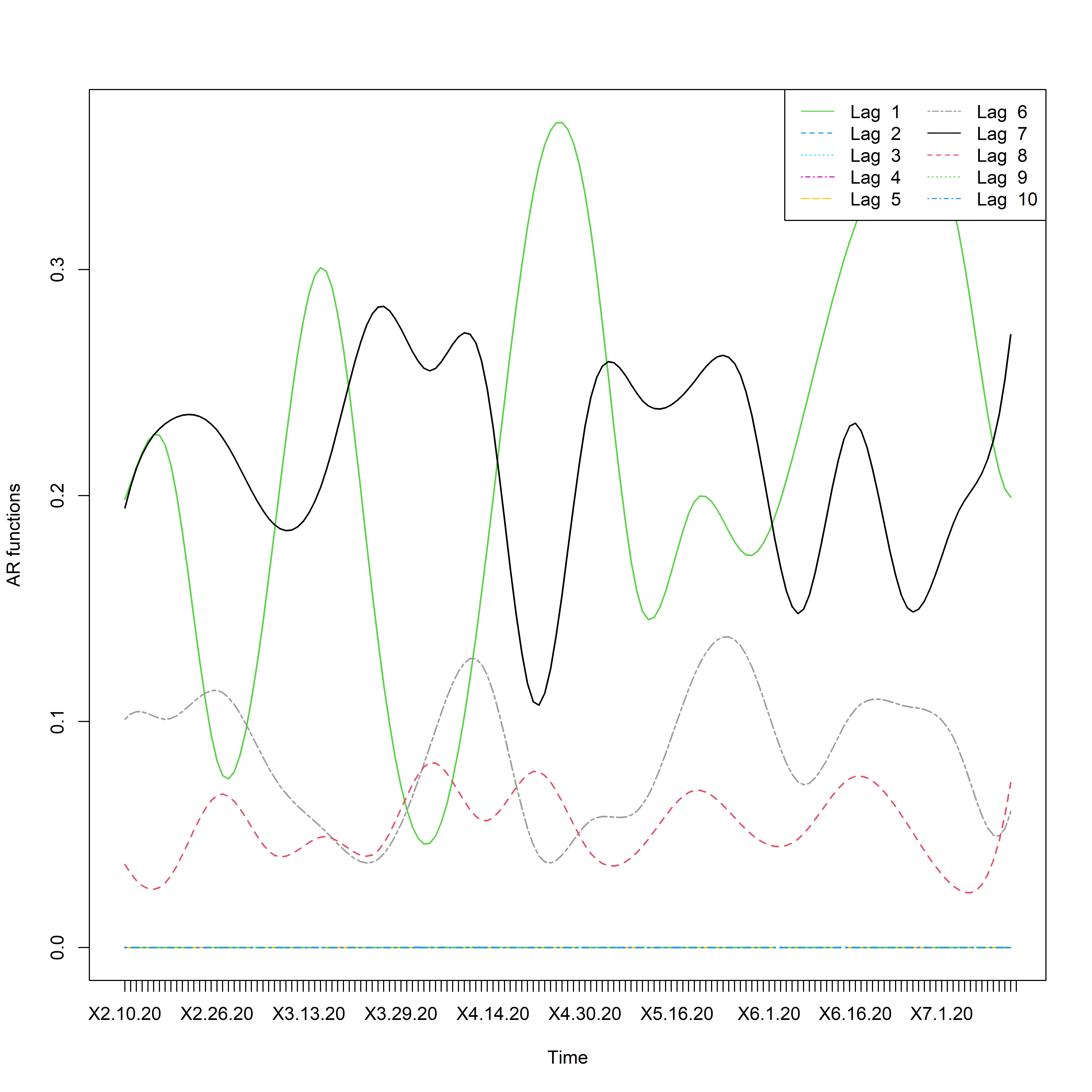}}
		\caption{Estimated mean functions in 1st column and estimated AR coefficient functions in the 2nd column for NYC using TVBARC(10). Black is the estimated curve along with the 95\% pointwise credible bands in green for the mean function.} 
		\label{fig:NYC10}
	\end{figure}
	
	\begin{figure}
		\centering
		\includegraphics[width=120mm]{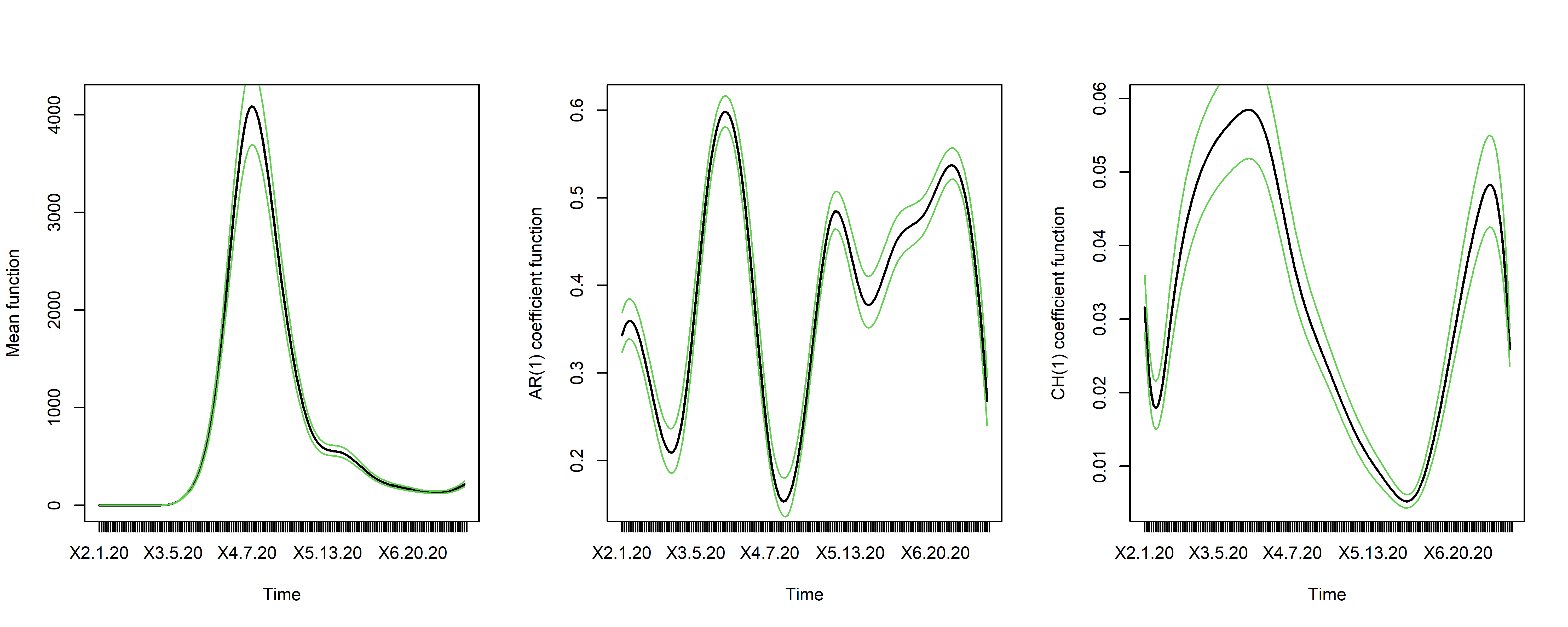}
		\caption{Estimated coefficient functions for the TVBINGARCH(1,1) on NYC data. Black is the estimated curve along with the 95\% pointwise credible bands in green.}
		\label{fig:NYCTVBINGARCH}
	\end{figure}

	\begin{table}[ht]
		\caption{Average MSE comparison for different methods on NYC data.}
		\centering
		\resizebox{.9\textwidth}{!}{
			\begin{tabular}{rrrrrr}
				Method & AMSE & Method & AMSE & Method & AMSE\\
				\hline
				INGARCH(1,1) & 318056.3 & GARMA(10,0) & 1682976.1 & TVBARC(1) & 210258.9\\
				GARMA(1,1) & 329610.1 & tvAR(1) & 338970.6 & TVBARC(10) & 185777.9\\
				AR(10,0) & 1376133.7 & tvAR(10) & 274913.7 & TVBINGARCH(1,1) & 212168.1\\
				\hline
			\end{tabular}
			\label{tab:NYCcompare}}
	\end{table}

	\section{Discussion}
	\label{discussion}
	We propose a time-varying Bayesian autoregressive model for counts (TVBARC) and time-varying Bayesian integer-valued generalized autoregressive conditional heteroskedastic model (TVBINGARCH) with linear link function within Poisson error to study the time series of daily new confirmed cases of COVID-19. We develop a novel hierarchical Bayesian model that satisfies the stability condition for the respective time-varying models and propose an HMC algorithm based MCMC sampling scheme. We also establish posterior contraction rate results of the proposed Bayesian methods. The `R' function with an example code can be found at {\url{https://github.com/royarkaprava/TVBARC}}. Relying on the proposed hierarchical Bayesian model, one can develop a time-varying Bayesian model for positive-valued time-series data too. Our analysis of NYC data shows that there is a time-varying effect of Lag 6, 7, and 8. Some preliminary analysis on COVID data using our model based on the data until April 24 are archived in our unpublished pre-print \cite{roy2020bayesian}. There are some more interesting findings related to significant lags for different countries.
	
	
	The definition of posterior contraction rate we followed involves an diverging sequence $M_T\rightarrow\infty$. If it is possible to replace $M_T$ by a large constant $M$ without changing $\epsilon_T$, the contraction rate then holds in a slightly stronger sense (see Chapter 8, \cite{ghosal2017fundamentals}). Local stationary approximation of the proposed nonstationary process is expected help to establish such result.
	Establishing Bernstein von-Mises type theorem to ensure asymptotic normality of the posterior distribution will also be interesting.
	However, such results are not yet available for the corresponding stationary cases.
	Nevertheless, it is also important direction of future research. 
	
	As future work, it will be interesting to include some country-specific information such as demographic information, geographical area, the effect of environmental time-series, etc in the model. These are usually important factors for the spread of any infectious disease. We can also categorize the different types of government intervention effects to elaborate more on the specific impacts of the same. In the future we wish to analyze the number of deaths, number of recovered cases, number of severe/critical cases, etc. for these diseases as those will hopefully have different dynamics than the one considered here and can provide useful insights about the spread and measures required. For computational ease, we have considered the same level of smoothness for all the coefficient functions. Fitting this model with different levels of smoothness might be able to provide more insights. Lag selection is a difficult task for time-varying auto-regressive models. One potential future direction would be to put sparsity inducing prior to the time-varying coefficient functions in TVBARC for automatic lag detection. Other than building time-varying autoregressive models for count-valued data using the hierarchical structure from this article, one interesting future direction is to extend this model for vector-valued count data. In general, it is difficult to model multivariate count data. There are only a limited number of methods to deal with multivariate count data \citep{besag1974spatial,yang2013poisson,roy2019nonparametric}. Building on these multivariate count data models, one can extend our time-varying univariate AR$(p)$ to a time-varying vector-valued AR$(p)$. On the same note, even though we imposed Poisson assumption for increased model interpretation, in the light of the upper bounds for the KL distance, it is not a necessary criterion and can be applied to a general multiple non-stationary count time-series. Extending some of the continuous time-series invariance results for nonlinear non-stationary and multiple series from \cite{karmakar2020optimal} to a count series regime will be an interesting challenge. Finally, we wish to undertake an autoregressive estimation of the basic reproduction number with the time-varying version of compartmental models in epidemiology.

	\section{Proof of Theorems}\label{sec:proof}
	We study the frequentist property of the posterior distribution is increasing $T$ regime assuming that the observations are coming from a true density $f_0$ characterized by the parameter $\kappa_0$. We follow the general theory of \cite{ghosal2000convergence} to study the posterior contraction rate for our problem. In the Bayesian framework, the density $f$ is itself a random measure and has distribution $\Pi$ which is the prior distribution induced by the assumed prior distribution on $\kappa$. The posterior distribution of a neighborhood $U_T=\{f:d(f,f_0)<\epsilon_T\}$ around $f_0$ given the observation $X^{(T)}=\{X_0,X_1,\ldots,X_T\}$ is
	$$
	\Pi_T(U_T^c|X^{(T)})=\frac{\int_{U_T^c}f(X^{(T)})d \Pi(\kappa)}{\int f(X^{(T)})d\Pi(\kappa)}
	$$

	\subsection{General proof strategy}\label{sec:genproof}
	
	The posterior consistency would hold if above posterior probability almost surely goes to zero in $F_{\kappa_0}^{(T)}$ probability as $T$ goes to $\infty$, where $F_{\kappa_0}^{(T)}$ is the true distribution of $X^{(T)}$. Recall the definition of posterior contraction rate; for a sequence $\epsilon_T$ if $\Pi_T(d(f,f_0)|X^{(T)}\geq M_T\epsilon_T|X^{(T)})\rightarrow 0$ in $F_{\kappa_0}^{(T)}$-probability for every sequence $M_T\rightarrow\infty$, then the sequence $\epsilon_T$ is called the posterior contraction rate. If the assertion is true for a constant $M_T=M$, then the corresponding contraction rate becomes slightly stronger. 
	
	Note that for two densities $f_{0}, f$ characterized by $\kappa_0$ and $\kappa$ respectively, the Kullback-Leibler divergences are given by 
	\begin{gather*}
	KL(\kappa_0, \kappa) = \int f_0\log{\frac{f_0}{f}}=E_{\kappa_0}\left[\log\frac{\IP_{Q_{\kappa_0}}(X_0)\prod_{t=1}^T\IP_{\kappa_0}(X_t|\sF_{t-1},\lambda_0)}{\IP_{Q_{\kappa}}(X_0)\prod_{t=1}^T\IP_{\kappa}(X_t|\sF_{t-1},\lambda_0)}\right].
	\end{gather*}
	
	\noindent Assume that there exists a sieve in parameter space such that $\Pi(W_T^c)\leq \exp(-(C_T+2)T\epsilon_T^2)$ and we have tests $\chi_T$ such that $$\IE_{\kappa_0}(\chi_T)\leq e^{-L_TT\epsilon_T^2/2}\quad \sup_{\kappa\in W_T: d^2(f,f_0)>L_{T}\epsilon_T^2}\IE_{\kappa}(1-\chi_T)\lesssim e^{-L_{T}T\epsilon_T^2}$$ for some $L_{T}>C_T+2$. Say $U_T=\{f:d^2(f,f_0)\leq L_T\epsilon_T^2\}$  and $S_T=\{\int \frac{f(X^T)}{f_0(X^T)}d\Pi(\kappa)\geq\Pi_{T}(\frac{1}{T}KL(\kappa_0,\kappa)<\epsilon_T)\exp(-C_{T}T\epsilon_T^2)\}$. We can bound the posterior probability from above by,
	\begin{align}
	\Pi_T(d(f,f_0)\geq M_T\epsilon_T|X^{(T)}) &\leq \chi_T+(1-\chi_T)\frac{\int_{U_T^c}f(X^T)d \Pi(\kappa)}{\int f(X^{(T)})d\Pi(\kappa)}\nonumber\\
	&=\chi_T+(1-\chi_T)\frac{\int_{U_T^c}\frac{f(X^{(T)})}{f_0(X^{(T)})}d \Pi(\kappa)}{\int \frac{f(X^{(T)})}{f_0(X^{(T)})}d\Pi(\kappa)}\nonumber\\
	&\leq \chi_T+\mathbb{1}\{S_T^c\} + (1-\chi_T)\frac{\int_{U_T^c}\frac{f(X^{(T)})}{f_0(X^{(T)})}d \Pi(\kappa)}{\exp(-C_TT\epsilon_T^2)\Pi_{T}\{\frac{1}{T}KL(\kappa_0,\kappa)<\epsilon_T\}} \nonumber\\
	&\leq \chi_T+\mathbb{1}\{S_T^c\} + \frac{\exp(C_TT\epsilon_T^2)}{\Pi_{T}\{\frac{1}{T}KL(\kappa_0,\kappa)<\epsilon_T\}}(1-\chi_T)\frac{\int_{U_T^c}f(X^{(T)})}{f_0(X^{(T)})}d \Pi(\kappa)
	\end{align}
	Taking expectation with respect to $\kappa_0$, first term go to zero by construction of $\chi_T$. The second term $\IE_{\kappa_0}\mathbb{1}\{S_T^c\}$ goes to zero due to Lemma 8.21 of \cite{ghosal2017fundamentals} for any sequence $C_T\rightarrow \infty$. We would require that $\Pi_{T}\{\frac{1}{T}KL(\kappa_0,\kappa)< \epsilon_T\} \geq \exp(-T\epsilon_T^2)$. Then for the third term,
	\begin{align}
	&\IE_{\kappa_0}\exp((C_T+1)T\epsilon_T^2)(1-\chi_T)\frac{\int_{U_T^c}f(X^{(T)})}{f_0(X^{(T)})}d \Pi(\kappa)=\exp((C_T+1)T\epsilon_T^2)\int_{U_T^c}f(X^{(T)})(1-\chi_T)d\Pi(\kappa)\nonumber\\
	&\quad\leq\exp(C_T+1)T\epsilon_T^2)\left[\int_{U_T^c\cap W_T}f(X^{(T)})(1-\chi_T)d\Pi(\kappa) + \Pi(W_T^c)\right]\nonumber\\&\quad=\exp((C_T+1)T\epsilon_T^2)\left[\sup_{\kappa\in W_T: d^2(f,f_0)>L_{T}\epsilon_T^2}\IE_{\kappa}(1-\chi_T) + \Pi(W_T^c)\right]\lesssim \exp(-T\epsilon_T^2).
	\end{align}
	
	\noindent Thus we need three things to calculate posterior contraction rate.
	\begin{itemize}
		\label{req}
		\item[(i)](Prior mass Condition) We would require $\Pi_{T}\{\frac{1}{T}KL(\kappa_0,\kappa)<\epsilon_T\} \geq \exp(-T\epsilon_T^2)$, 
		\item[(ii)](Sieve) construct the sieve $W_T$ such that $\Pi(W_T^c)\leq \exp(-(C_T+2)T\epsilon_T^2)$ and
		\item[(iii)](Test construction) exponentially consistent tests $\chi_T$.
	\end{itemize}
	We first study the contraction properties with respect to $d^2(f,f_0)=r_T^2(f, f_0)=-\frac{1}{T}\log\int\sqrt{ff_0}$ and then show that the same rate holds for average Hellinger $\frac{1}{T}d_H^2(f,f_0)$. Note that $L_T$ can be taken as $L_T=M_T^2$. With the above general structure, we now proceed to prove individual theorems focusing on the TVBARC and the TVINGARCH cases.

	\ignore{
		The requirement that $\sup_{\kappa_0\in\mathcal{A}}\{F_{\kappa_0}(S_T^c)\}\rightarrow 0$ where $\mathcal{A}$ stands for the true parameter space satisfying the conditions in Assumption 1-3 is ensured if we have 
		\begin{align}
		-\log\Pi\left\{\kappa:KL(\kappa_0,\kappa)\leq T\epsilon_T^2, V(\kappa_0,\kappa)\leq T\epsilon_T^2\right\}\leq T\epsilon_T^2\label{eq:priormasscondi}
		\end{align} due to Lemma 8.10 of \cite{ghosal2017fundamentals}. Also in that case, we get slightly stronger rate with some fixed large constant $M_T=M$.
	}

	\subsection{Proof of Theorem 1}\label{sec:proofthm1}
	For the sake of technical convenience we show our proof for time-varying AR model with 1 lag only. All the proofs go through for higher lags with the same technical tools.
	\subsubsection{KL Support}\label{sec:klbound}

	The likelihood based on the parameter space $\kappa$ is given,
	$
	\IP_{\kappa}(X_0)\prod_{t=1}^T\IP_{\kappa}(X_t|X_{t-1}).
	$ Let $Q_{\kappa, t}(X_t)$ be the distribution of $X_t$ with parameter space $\kappa$.

	We have 
	\begin{align}
	&R=\log\frac{\prod_{t=1}^T\IP_{\kappa_0}(X_t|\sF_{t-1},\lambda_0)}{\prod_{t=1}^T\IP_{\kappa}(X_t|\sF_{t-1},\lambda_0)}\nonumber\\
	&\quad=\sum_{t=1}^T[-\{\mu_0(t/T)-\mu(t/T)\} -\{a_{01}(t/T)-a_{1}(t/T)\}X_{t-1} + X_t\{\log(\mu_0(t/T)+a_{01}(t/T)X_{t-1})\nonumber\\
	&\qquad-\log(\mu(t/T)+a_{1}(t/T)X_{t-1})\} ]\label{eq:KLR}
	\end{align}
	Then $KL(\kappa_0, \kappa)=\IE_{\kappa_0}(R)$. We have in light of MVT,
	\begin{align}
	&|R|\leq\sum_{t=1}^T[|\mu(t/T)-\mu_0(t/T)|+ |a_1(t/T)-a_{01}(t/T)|X_{t-1}\\&\quad+\frac{X_t}{\mu_{*}(t/T)+a_{1*}(t/T)X_{t-1}}\{|\mu(t/T)-\mu_0(t/T)|+ |a_1(t/T)-a_{01}(t/T)|X_{t-1}\}]\nonumber\\
	&\qquad\leq T\|\mu-\mu_0\|_{\infty}+\|a_1-a_{01}\|_{\infty}\sum_{t}X_{t-1} + \|\mu-\mu_0\|_{\infty}/\rho\sum_t X_t + \|a_1-a_{01}\|_{\infty}/\rho\sum_t X_t,
	\end{align}
	under the assumption that $\kappa(\cdot)=(\mu(\cdot),a_1(\cdot), b_1(\cdot))$ and $\kappa_0(\cdot)=(\mu_0(\cdot),a_{10}(\cdot),b_{10}(\cdot))$ are close and also $\kappa_{*}$ is close to both and also in conjunction with Assumption (A.3) to imply $\inf_t a_{1*}(t/T) >\rho$ and Assumption (A.2) which implies $ \IE(X_t) <M_X.$ Then for the first term we use the bound $\mu_{*}(t/T)+a_{1*}(t/T)X_{t-1}>\rho$ and for the second term the bound $\mu_{*}(t/T)+a_{1*}(t/T)X_{t-1}>\rho X_{t-1}$ is used to have $\frac{|\mu(t/T)-\mu_0(t/T)|+ |a(t/T)-a(t/T)|X_{t-1}}{\mu_{*}(t/T)+a_{1*}(t/T)X_{t-1}}\leq \|\mu-\mu_0\|_{\infty}/\rho + \|a_1-a_{01}\|_{\infty}/\rho$ for all $t$. Thus,
	\begin{align}
	\frac{1}{T}\IE(R)\lesssim \|\mu-\mu_0\|_{\infty}+\|a_1-a_{01}\|_{\infty}.\label{eq:KLineq1}
	\end{align}

	\subsubsection{Posterior contraction in terms of average negative log-affinity}
	
	In this section, we focus on the requirements to calculate posterior contraction rate as in Section~\ref{req}.
	We first show posterior consistency in terms of average negative log-affinity which is defined as $r_T^2(f_1,f_2)=-\frac{1}{T}\log\int f_1^{1/2}f_2^{1/2}$ between $f_1$ and $f_2$. Here, we have $f_1=\prod_{i=1}^T P_{\kappa_1}(X_{i}|X_{i-1})$. Then we show that, having $r_T^2(f_1,f_0)\lesssim\epsilon_n^2$ implies that our distance metric $d_{2,T}^2(f_1,f_0)\lesssim \epsilon_n^2$.

	Proceeding with the rest of the proof of Theorem 1, we use the results of B-Splines,
	$\|\mu-\mu_0\|_{\infty}\leq \|\alpha-\alpha_0\|_{\infty},$ where $\alpha = \{\alpha_j\}$ and $\|a_1-a_{10}\|_{\infty}\leq \|\gamma-\gamma_{0}\|_{\infty}$, where $\gamma_{j}=\theta_{1j}M_1,$ such that $\gamma_{j}<1$. The H\"older smooth functions with regularity $\iota$ can be approximately uniformly up to order $K^{-\iota}$ with $K$ many B-splines. Thus we have $\epsilon_T\gtrsim\max\{K_{1T}^{-\iota_1},K_{2T}^{-\iota_2}\}$. 
	
	We need to lower bound the prior probability as required by (i). We have the result~\eqref{eq:KLineq1} and the prior probabilities
	$\Pi(\|\alpha-\alpha_0\|_{\infty}\lesssim\epsilon_T, \|\gamma-\gamma_{0}\|_{\infty}\lesssim\epsilon_T)\gtrsim \epsilon_T^{K_{1T}+K_{2T}}$ based on the discussion of A2 from \cite{shen2015adaptive}. The rate of contraction cannot be better than the parametric rate $T^{-1/2}$, and so $\log (1/\epsilon_T)\lesssim \log T$. Thus (i) requires that in terms of pre-rate $\bar{\epsilon}_T$, we need $(K_{1T}+K_{2T})\log T\lesssim T\bar{\epsilon}_T^2$. 
	
	In our problem, we consider following sieve as required by (ii) \begin{align}
	W_T&=\{K_1,K_2,\alpha,\gamma:K_1\leq K_{1T},K_2\leq K_{2T}, \|\alpha\|_{\infty}\leq A_{T}, \min(\alpha,\gamma)>\rho_T, \gamma\leq 1-A_T/B_T, \nonumber\\&\qquad\lambda_0\leq B_T, A_T<B_T \},
	\end{align}
	where $A_T, B_T$ are at least polynomial in $T$ and $\lambda_0$ is the mean of $X_0$ and $K_T=\max\{K_{1T},K_{2T}\}$. We take $\rho_T\asymp T^{-a}$ with $a<1$, $A_T\asymp T^{a_1}, B_T\asymp T^{a_2}$ with $a_2>a_1$ for technical need. Note that, for $\kappa\in W_T$, we have $\IE_{\kappa}(X_t)<B_T$. We need to choose these bounds carefully so that we have $\Pi(W_T^c)\leq \exp(-(1+C_1)T\epsilon_T^2)$, which depend on tail properties of the prior. We have, $\Pi(W_T^c)= \Pi [K_1>K_{1T}, K_2>K_{2T}, \alpha_{K_{1T}}\in\{x:\inf x > \rho_T, \sup x<  A_T, \gamma_{K_{2T}}\notin\{x:\inf x>\rho_T, \sup x<1-\frac{A_T}{B_T}\}, \lambda_0>B_T]$.
	
	Hence we have,
	$\Pi(W_T^c)\leq\Pi(K_1>K_{1T})+\Pi(K_2>K_{2T})+\Pi\{\alpha_{K_{1T}}\notin[\rho_T, A_T]^{K_{1T}}\}+\Pi\{\gamma_{K_{2T}}\notin\left[\rho_T, 1-\frac{A_T}{B_T}\right]^{K_{2T}}\}+\Pi\{\lambda_0>B_T\}$ where $\alpha_{K_{1T}}$ is the vector of full set of coefficients of length $K_{1T}$ and $\gamma_{K_{2T}}$ is the vector of coefficients of length $K_{2T}$. The quantity $\Pi[\alpha_{K_{1T}}\notin[\rho_T, A_T]^{K_{1T}}$ can be further upper bounded by $K_{1T}\Pi(\alpha_{1}\notin[\rho_T, A_T)])\leq K_{1T}\exp\{-R_1T^{a_3}\}$, for some constant $R_1, a_3>0$ which can be verified from the discussion of the assumption A.2 of \cite{shen2015adaptive} for our choice of prior which exponential. On the other hand, $\Pi\{\gamma_{K_{2T}}\notin\left[\rho_T, 1-\frac{A_T}{B_T}\right]^{K_{2T}}\leq K_{2T}\Pi(\gamma_1\notin[\rho_T, 1-\frac{A_T}{B_T}])\leq K_{2T}\exp\{-R_2T^{a_4}\}$ for some constant $R_2, a_4>0$ which can be verified from the proof of \cite{roy2018high}. The inverse-gamma prior of $\lambda_0$ has exponential tail similar to $\alpha_1$ and thus can be ignored as $K_{1T}$ grows with $T$. Since $B_T>A_T$, the tail of $\lambda_0$ can be upper bounded by tail of $\alpha_1$  

	Hence, $\Pi(W_T^c)\lesssim F_1(K_{1T})+F_2(K_{2T})+(K_{1T}+K_{2T})\exp\{-RT^{a_5}\}$.  The two functions $F_1$ and $F_2$ in the last expression stand for the tail probabilities of the prior of $K_1$ and $K_2$. We can calculate their asymptotic order as, $F_1(x)=\Pi(K_{1}>x)\asymp\exp\{-x(\log x)^{b_{13}}\}$ and $F_2(x)=\Pi(K_{2}>x)\asymp\exp\{-x(\log x)^{b_{23}}\}$. We need $\Pi(W_T^c)\lesssim\exp\{-(1+C_T)T\epsilon_T^2\}$. Hence, we calculate pre-rate from the following equation for some sequence $H_T\rightarrow\infty$,
	
	\begin{align}
	\label{eq:rate sieve}
	{K}_{1T} (\log T)^{b_{13}}+{K}_{2T}(\log T)^{b_{23}}\gtrsim H_TT\bar\epsilon_T^2, \quad \log(K_{1T}+K_{2T})+H_TT\bar\epsilon_T^2\lesssim T^{a_5}.
	\end{align}

	Now, we construct test $\chi_T$ such that $$\IE_{\kappa_0}(\chi_T)\leq e^{-L_TT\epsilon_T^2/2}\quad \sup_{\kappa\in W_T: r_T^2(\kappa,\kappa_0)>L_T\epsilon_T^2}\IE_{\kappa}(1-\chi_T)\lesssim e^{-L_TT\epsilon_T^2}$$ for some $L_T>C_T+2$. 
	
	To construct the test as required in (iii), we first construct the test for point alternative $H_0:\kappa=\kappa_0$ vs $H_1:\kappa=\kappa_1$. The most powerful test for such problem is Neyman-Pearson test $\phi_{1T}=\mathbb{1}\{f_1/f_0\geq 1\}$. For $r_T^2> L_T\epsilon_T^2$, we have
	
	$$\IE_{\kappa_0}\phi_{1T}=\IE_{\kappa_0}(\sqrt{f_1/f_0}\geq 1)\leq \int \sqrt{f_1f_0}\leq \exp(-L_TT\epsilon_T^2),$$ 
	
	$$\IE_{\kappa_1}(1-\phi_{1T})=\IE_{\kappa_1}(\sqrt{f_0/f_1}\geq 1)\leq \int \sqrt{f_0f_1}\leq \exp(-L_TT\epsilon_T^2).$$
	
	It is natural to have a neighborhood around $\kappa_1$ such the Type II error remains exponentially small for all the alternatives in that neighborhood under the test function $\phi_{1T}$. By Cauchy-Schwarz inequality, we can write that $$\IE_{\kappa}(1-\phi_{1T})\leq \{\IE_{\kappa_1}(1-\phi_{1T})\}^{1/2}\{\IE_{\kappa_1}(f/f_1)^2\}^{1/2}.$$In the above expression, the first factor already exponentially decaying. The second factor can be allowed to grow at most of order $e^{cT\epsilon_T^2}$ for some positive small constant $c$. We show that $\IE_{\kappa_1}(f/f_1)^2$ is bounded for every $\kappa$ such that
	$$\|\mu-\mu_1\|_{\infty}\leq\frac{\sqrt{\rho_T}}{\sqrt{T}}, \|a-a_1\|_{\infty}\leq\frac{\sqrt{\rho_T}}{\sqrt{TB_T}}. $$
	
	\noindent We have, in the light of AM-GM inequality, $$\IE_{\kappa_1}(f/f_1)^2=\int\frac{f^2}{f_1^2}f_1=\int\frac{f}{f_1}f=\IE_{\kappa}\frac{f}{f_1}=\IE_{\kappa}\prod_{t=1}^T\frac{f(X_t|X_{t-1})}{f_1(X_t|X_{t-1})}\leq\frac{1}{T}\sum_{t=1}\IE_{\kappa}\left(\frac{f(X_t|X_{t-1})}{f_1(X_t|X_{t-1})}\right)^T$$
	\ignore{
		For $p=\textrm{Poi}(\lambda)$ and $p_1=\textrm{Poi}(\lambda_1)$. we have $\IE_{p_1}(p/p_1)^2=\sum_{y=1}^{\infty}\frac{1}{y!}e^{\lambda_1-2\lambda}\lambda^{2y}/\lambda_{1}^y=\exp[\frac{(\lambda-\lambda_1)^2}{\lambda_1}]$. Take $\lambda(t/T)=\mu(t/T)+a_1(t/T)X_{t-1}$ and $\lambda_1=\mu_1(t/T)+a_{11}(t/T)X_{t-1}$. For the parameters in $W_T$, we have $\lambda(t/T)>\rho_T$ or $X_{t-1}\rho_T$. Thus, $\exp[\frac{(\lambda-\lambda_1)^2}{\lambda_1}]\leq\exp[\frac{2\{\mu(t/T)-\mu_1(t/T)\}^2+2\{a_1(t/T)-a_{11}(t/T)\}^2X_{t-1}^2}{\lambda_1}]$

		For $p=\textrm{Poi}(\lambda)$ and $p_1=\textrm{Poi}(\lambda_1)$. we have $\IE_{p_1}(p/p_1)^2=\sum_{y=1}^{\infty}\frac{1}{y!}e^{\lambda_1-2\lambda}\lambda^{2y}/\lambda_{1}^y=\exp[\frac{(\lambda-\lambda_1)^2}{\lambda_1}]$ and $\int \frap{p}/{p_1}=\exp[(\lambda-\lambda_1) + \frac{\lambda}{\lambda_1}]\leq \exp[|\lambda-\lambda_1|(1+\frac{1}{\lambda_1}) + 1]$. For the parameters in $W_T$, we have $\lambda(t/T)>\rho_T$. 
		
		Using this result we get that,
		\begin{align}
		&\IE_{\kappa_1}(f/f_1)^2= \int f^2/f_1
		\end{align}

		\noindent Next, use AM-GM inequality $\IE_{f}\frac{f}{f_1}=\IE_{f}\prod_{t=1}^T\frac{f(X_t|X_{t-1})}{f_1(X_t|X_{t-1})}\leq\frac{1}{T}\sum_{t=1}\IE_{f}\left(\frac{f(X_t|X_{t-1})}{f_1(X_t|X_{t-1})}\right)^T$
		\begin{align}
		&\IE_{\kappa_1}(f/f_1)^2= \IE_{f}\frac{f}{f_1} \leq\frac{1}{T}\sum_{t=1}^T\IE_{f} \left(\frac{f(X_t|X_{t-1})}{f_1(X_t|X_{t-1})}\right)^T\leq\sup_t\IE_{f} \left(\frac{f(X_t|X_{t-1})}{f_1(X_t|X_{t-1})}\right)^T
		\end{align}

	}
	
	\noindent Towards uniformly bounding the summand in the above display, we  write \begin{align}
	\IE_{\kappa} \left(\frac{f(X_t|X_{t-1})}{f_1(X_t|X_{t-1})}\right)^T&=\IE_{X_{t-1},\kappa}\sum_{X_t=0}^{\infty}\frac{\{f(X_t|X_{t-1})\}^{T}}{\{f_1(X_t|X_{t-1})\}^T}f(X_t|X_{t-1})\nonumber\\
	&=\IE_{X_{t-1},\kappa}\exp[-T(\lambda-\lambda_1)-\lambda]\sum_{X_t=0}^{\infty}\left(\frac{\lambda^{T+1}}{\lambda_1^T}\right)^{X_t}/X_t! \nonumber \\
	&=\IE_{X_{t-1},\kappa}\exp[-T(\lambda-\lambda_1)-\lambda+\frac{\lambda^{T+1}}{\lambda_1^T}] \nonumber \\
	&=\IE_{X_{t-1},\kappa}\exp\Bigg[-T\{\mu(t/T)-\mu_1(t/T)\}-T\{a_1(t/T)-a_{11}(t/T)\}X_{t-1}\nonumber\\
	&\quad\quad\quad\quad\quad -\mu(t/T) - a_1(t/T)X_{t-1}+\frac{(\mu(t/T)+a_1(t/T)X_{t-1})^{T+1}}{(\mu_1(t/T)+a_{11}(t/T)X_{t-1})^{T}}\Bigg]\label{eq:tobdexp}.
	\end{align}
	where, $\lambda=\mu(t/T)+a_1(t/T)X_{t-1}$,$\lambda_1=\mu_1(t/T)+a_{11}(t/T)X_{t-1}$ and $\IE_{X_{t-1},\kappa}$ denotes unconditional expectation over $X_{t-1}$ under the density $f$ with parameter $\kappa$. Let us define $r_1=\|\mu-\mu_1\|_{\infty}$ and $r_2=\|a_1-a_{11}\|_{\infty}$

	Assuming $\mu(t/T)-\mu_1(t/T)$ and $a_1(t/T)-a_{11}(t/T)$ very small, we can write
	\begin{align}
	&\left[ \frac{(\mu(t/T)+a_1(t/T)X_{t-1})^{T+1}}{(\mu_1(t/T)+a_{11}(t/T)X_{t-1})^{T}}\right]\nonumber\\
	&=\left\{1+\frac{\mu(t/T)-\mu_1(t/T)+(a_1(t/T)-a_{11}(t/T))X_{t-1}}{\mu_1(t/T)+a_{11}(t/T)X_{t-1}}\right\}^T(\mu(t/T)+a_1(t/T)X_{t-1})\nonumber\\
	&\approx\left\{1+T\frac{\mu(t/T)-\mu_1(t/T)+(a_1(t/T)-a_{11}(t/T))X_{t-1}}{\mu_1(t/T)+a_{11}(t/T)X_{t-1}}\right\}(\mu(t/T)+a_1(t/T)X_{t-1})\label{eq:impeq}
	\end{align}
	For the above approximation to hold, we need $\frac{\mu(t/T)-\mu_1(t/T)+(a_1(t/T)-a_{11}(t/T))X_{t-1}}{\mu_1(t/T)+a_{11}(t/T)X_{t-1}}$ to be small. To verify that, observe that $$\left|\frac{\mu(t/T)-\mu_1(t/T)+(a_1(t/T)-a_{11}(t/T))X_{t-1}}{\mu_1(t/T)+a_{11}(t/T)X_{t-1}}\right|\leq \frac{r_1}{\rho_T}+\frac{r_2}{\rho_T}=\frac{1}{\sqrt{T\rho_T}}(1+\frac{1}{\sqrt{B_T}}).$$ As we have $\rho_T=T^{-a}$ with $a<1$, it follows directly. 
	Thus ~\eqref{eq:tobdexp} before $\IE_{X_{t-1},\kappa}$ applying on ~\eqref{eq:impeq} becomes 
	\begin{align}
	&\exp\left[ \frac{[T\{\mu(\frac{t}{T})-\mu_1(\frac{t}{T})\}+T\{a_1(\frac{t}{T})-a_{11}(\frac{t}{T})\}X_{t-1}][\{\mu(\frac{t}{T})-\mu_1(\frac{t}{T})\}+\{a_1(\frac{t}{T})-a_{11}(\frac{t}{T})\}X_{t-1}]}{\mu_1(\frac{t}{T})+a_{11}(\frac{t}{T})X_{t-1}}\right]\nonumber\\
	&\leq \exp[Tr_1^2/\rho_T+2Tr_1r_2/\rho_T+Tr_2^2X_{t-1}/\rho_T]\label{eq:exbd}
	\end{align}
	
	\noindent The bound in~\eqref{eq:exbd} is obtained by applying a combination of the following inequalities $\mu(t/T)+a_1(t/T)X_{t-1}>\rho_{T}$ or $>\rho_{T} X_{t-1}$, $|\mu(t/T)-\mu_1(t/T)|<r_1$ and $ |a_1(t/T)-a_{11}(t/T)|<r_2$. Taking $q=Tr_2^2/\rho_T$, last part becomes $\IE(e^{qX_{t-1}})$ after taking exectation over \eqref{eq:exbd}. We have $\IE(e^{qX_0})=e^{\lambda_0(e^q-1)}<e^{B_T(e^q-1)}=e^{Q}$ for $Q=B_T(e^q-1)\implies (e^q-1) = Q/B_T,$ $B_T$ is the upper bound for $\lambda_0$ in the sieve). We will show $\IE(e^{qX_1})<Q$ under the above choice of $r_1$ and $r_2$. Then by recursion it holds for all $t$. We use the result $e^q-1\leq 2q$ for $q<1$. 
	
	\ignore{
		\begin{eqnarray}
		\IE(e^{qX_1})&=&\IE_{X_0}\sum_{X_1=0}^{\infty}e^{qX_1}\exp\{-\mu(1)-a_1(1)X_{0}\}(\mu(1)+a_1(1)X_{0})^{X_1}/X_1!\\&=&\IE_{X_0}\exp\{(\mu(1)+a_1(1)X_0)(e^q-1)\}\\&=&\exp\{\mu(1)(e^q-1)\}\IE_{X_0}\exp\{a_1(e^q-1)X_0\}\\&=&\exp\{\mu(1)(e^q-1)+\lambda_0(e^{(e^q-1)a_1(1)-1})\}\\&=&\exp\{\mu(1)Q/B_T+\lambda_0(e^{Qa_1(1)/B_T}-1)\}.
		\end{eqnarray}
	}
	\noindent With $\lambda_1(X_0)= \mu(1)+ a_1(1)X_0$, we have $$\IE( e^{qX_1} )= \IE( \IE( e^ {qX_1}|X_0)) = E ( e^{\lambda_1(X_0) (e^q-1)})=e^{(e^q-1)\mu(1)}e^{\lambda_0(e^{(e^q-1)a_1(1)}-1)}$$ 
	
	\noindent Then choose sieve parameters such that $Qa_1/B_T=a_1(1)(e^q-1)\leq 2a_{1}(1)q$ is very small which is ensured as $q$ is very small. Then $\mu(1)Q/B_T+\lambda_0(e^{Qa_1(1)/B_T}-1)\approx Q\mu(1)/B_T+\lambda_0(\frac{Qa_1(1)}{B_T}) \leq Q\{\mu(1)/B_T + a_1(1)\}<Q$ as within the sieve $\mu(1)/B_T + a_1(1) < A_T/B_T + (1-A_T/B_T)=1$. Hence, $\IE(e^{qX_1}) < e^Q$. Recursively, for all $t$, we can show $\IE(e^{qX_t})<e^Q$.
	
	Our primary goal of showing $\IE_{\kappa_1}(f/f_1)^2<\infty$ can be fulfilled if $Q$ is a constant, independent of $T$. To ensure $Q$ is independent of $T$ we need $B_T(e^q-1)$ is constant. It suffices to make $qB_T$ constant as $qB_T<B_T(e^q-1)<2qB_T$. Thus, for $r_2\leq\frac{\sqrt{\rho_T}}{\sqrt{TB_T}}$ and in the light of  \eqref{eq:exbd} $r_1\leq\frac{\sqrt{\rho_T}}{\sqrt{T}}$ we have $\IE_{\kappa_1}\left(\frac{f}{f_1}\right)^2$ bounded.

	The test function $\chi_T$ satisfying exponentially decaying Type I and Type II probabilities is then obtained by taking maximum over all tests $\phi_{jT}$'s for each ball, having above radius. Thus $\chi_T=\max_j\phi_{jT}$. Type I and Type II probabilities are given by $P_0(\chi_T)\leq \sum_{j}P_0\phi_{jT}\leq D_{T}P_0\phi_{jT}$ and $\sup_{\kappa\in W_T: r_T^2(\kappa,\kappa_0)>L_T\epsilon_T^2}P(1-\chi_T)\leq \exp(-TL_T\epsilon_T^2)$. Hence, we need to show that $\log D_{T}\lesssim T\epsilon_T^2$, where $D_T$ is the required number of balls of above radius needed to cover our sieve $W_T$. We have 
	\begin{align}
	\log D_T&\leq \log D(r_1, \|\alpha\|_{\infty}\leq A_T, \min(\alpha)>\rho_T, \|\cdot\|_{\infty})+\log D(r_2, \|\gamma\|_{\infty}\leq 1-\frac{A_T}{B_T}, \min(\gamma)>\rho_T, \|\cdot\|_{\infty}) \nonumber\\
	&\leq K_{1T}\log(3K_{1T}A_T/r_1)+K_{2T}\log(3K_{2T}/r_2)
	\end{align}
	
	\noindent Given our choices of $A_T, B_T$ and $\rho_T$, the two radii $r_1$ and $r_2$ are some fractional polynomials in $T$. Thus $\log D_T\lesssim (K_{1T}+K_{2T})\log T$, which is required to be $\lesssim T\epsilon_T^2$ as in the prior mass condition due to (i).
	
	Based on~\eqref{eq:rate sieve}, we have $\bar{K}_{1T}\asymp T^{1/(2\iota_1+1)}(\log T)^{-1/(2\iota_1+1)}$, $K_{2T}\asymp T^{1/(2\iota_2+1)}(\log T)^{-1/(2\iota_2+1)}$ and a pre-rate $\bar{\epsilon}_T=\max\bigg\{T^{-\iota_1/(2\iota_1+1)} (\log T)^{\iota_1/(2\iota_1+1)},T^{-\iota_2/(2\iota_2+1)} (\log T)^{\iota_2/(2\iota_2+1)}\bigg\}$. The actual rate will be slower that pre-rate. Now, the covering number condition, prior mass conditions and basis approximation result give us $(K_{1T}+K_{2T})\log T\lesssim T\epsilon_T^2$ and $\epsilon_T\gtrsim\max\{K_{1T}^{-\iota_1},K_{2T}^{-\iota_2}\}$.
	Combining all these conditions, we would require $K_{1T}\asymp T^{1/(2\iota_1+1)}(\log T)^{2\iota_1/(2\iota_1+1)-b_{13}}$, $K_{2T}\asymp T^{1/(2\iota_2+1)}(\log T)^{2\iota_2/(2\iota_2+1)-b_{23}}$. Hence we calculate the posterior contraction rate as $\epsilon_T$ equal to \begin{align*}
	\max\bigg\{&T^{-\iota_1/(2\iota_1+1)} (\log T)^{\iota_1/(2\iota_1+1)+(1-b_{13})/2},T^{-\iota_2/(2\iota_2+1)} (\log T)^{\iota_2/(2\iota_2+1)+(1-b_{23})/2}\bigg\}.
	\end{align*}
	\subsubsection{Posterior contraction in terms of average Hellinger}
	We can write Reyni divergence as $r_T^2=-\frac{1}{T}\log\int\sqrt{f_0f_1}=-\frac{1}{T}\log \IE_{\kappa_0}\sqrt{\frac{f_1}{f_0}}$. We need to show $r_T^2\lesssim\epsilon_T^2$ implies that $d_{2,T}^2(\kappa_0,\kappa)\lesssim \epsilon_T^2$ as $\epsilon_T$ goes to zero. 
	
	If $r_T^2\leq\epsilon_T^2$, we have $\left(\IE_{\kappa_0}\sqrt{\frac{f_1}{f_0}}\right)^{-1/T}\leq\exp(\epsilon_T^2)$ which implies for small $\epsilon_T^2$, we have $\left(\IE_{\kappa_0}\sqrt{\frac{f_1}{f_0}}\right)^{1/T}\geq 1-\epsilon_T^2$. By Cauchy-Squarz inequality $\left(\int\sqrt{f_0f_1}\right)^2\leq \int f_0\int f = 1$. Thus we have,
	
	$$
	1-\epsilon_T^2\leq \left(\IE_{\kappa_0}\sqrt{\frac{f_1}{f_0}}\right)^{1/T}\leq 1,
	$$
	
	Since $d_{H}^2(f_1,f_0)=2(1-\IE_{\kappa_0}\sqrt{\frac{f_1}{f_0}})$
	$$\left(\IE_{\kappa_0}\sqrt{\frac{f_1}{f_0}}\right)^{1/T}=\left\{1-\left(1-\IE_{\kappa_0}\sqrt{\frac{f_1}{f_0}}\right)\right\}^{1/T}\approx1-\frac{1}{2T}d_{H}^2(f_1,f_0).$$ Thus $\frac{1}{T}d_{H}^2(f_1,f_0)\lesssim\epsilon_T^2$. Thus it is consistent under average Hellinger distance.

	\subsection{Proof of Theorem 2}\label{sec:proofthm2}
	
	The proof will follow similar path as in the previous section. Thus we just specifically touch upon the parts that require different treatment. We can rewrite history of the INGARCH process as $\{\sF_{t-1},\sG_{t-1}\}=\{\sF_{t-1},\lambda_0\}$. For the INGARCH case, the likelihood based on the parameter space $\kappa$ is different from above and is given by,
	$
	\IP_{\psi_0}(X_0,\lambda_0)\prod_{t=1}^T \IP_{\psi}(X_t|\sF_{t-1},\lambda_0).
	$ Since all the steps are similar for the proof of Theorem 2, we only provide a outline. First to bound KL by the sup norm distances among functions, we need to tackle $|b_{11}(t/T)\lambda_{1t}-b_{01}(t/T)\lambda_{0t}|$. For this term we have 
	\begin{eqnarray}\label{eq:blambda}
	|b_{11}(t/T)\lambda_{1t}-b_{01}(t/T)\lambda_{0t}|\leq \lambda_{0t}\|b_{11}-b_{01}\|_{\infty}+\max_t{b_{11}(t)}|\lambda_{1t}-\lambda_{0t}|.
	\end{eqnarray}
	When $\psi_1$ is near $\psi_0$, we have for all $t$
	\begin{align*}
	|\lambda_{1t}-\lambda_{0t}|&\leq\|\mu_1-\mu_{0}\|_{\infty} + X_{t-1}\|a_{11}-a_{01}\|_{\infty}+(1-\frac{M_{\mu}}{M_{X}})|\lambda_{1,t-1}-\lambda_{0,t-1}| + \lambda_{0,t-1}|b_{11}-b_{01}|_{\infty}
	\end{align*}
	as we can upper bound $\max_t{b_{11}(t)}$ by $(1-\frac{M_{\mu}}{M_{X}})$ since $\psi_1$ is close to $\psi_0$. We have
	\begin{align*}
	&\sum_{t=1}^{T-1}\frac{M_{\mu}}{M_{X}}|\lambda_{1t}-\lambda_{0t}|+|\lambda_{1T}-\lambda_{0T}|\\&\quad\leq T\|\mu_1-\mu_{0}\|_{\infty} + \sum_{t}X_{t-1}\|a_{11}-a_{01}\|_{\infty} +(1-\frac{M_{\mu}}{M_{X}})|\lambda_{10}-\lambda_{00}|+ \sum_t\lambda_{0,t-1}|b_{11}-b_{01}|_{\infty}
	\end{align*}
	
	As $M_{\mu}<M_X$,
	\begin{align*}
	\sum_{t=1}^T|\lambda_{1t}-\lambda_{0t}|\leq & \frac{M_X}{M_{\mu}}\LARGE\{T\|\mu_1-\mu_{0}\|_{\infty} + \sum_{t}X_{t-1}\|a_{11}-a_{01}\|_{\infty} \\&\quad+(1-\frac{M_{\mu}}{M_{X}})|\lambda_{10}-\lambda_{00}|+ \sum_t\lambda_{0,t-1}|b_{11}-b_{01}|_{\infty}\LARGE\}.
	\end{align*}
	which implies,
	\begin{align}
	\IE\sum_{t=1}^T|\lambda_{1t}-\lambda_{0t}|\leq & \frac{M_X}{M_{\mu}}\LARGE\{T\|\mu_1-\mu_{0}\|_{\infty} + TM_X\|a_{11}-a_{01}\|_{\infty}\nonumber \\&\quad+(1-\frac{M_{\mu}}{M_{X}})|\lambda_{10}-\lambda_{00}|+ TM_X|b_{11}-b_{01}|_{\infty}\LARGE\}\label{eq:KLbdIN}.
	\end{align}
	Using the definition of $R$ as in~\eqref{eq:KLR}, we have
	\begin{align}
	&|R|\leq\sum_{t=1}^T\left[|\lambda_{1t}-\lambda_{0t}|+\frac{X_t}{\mu_{*}(t/T)+a_{1*}(t/T)X_{t-1}}|\lambda_{1t}-\lambda_{0t}|\right]
	\end{align}
	The first part follows directly. For the second part as $\psi_1$ and $\psi_0$ are close $$\sum_t\IE\left(\IE\left(\frac{X_t}{\lambda_{*t}}|\lambda_{1t}-\lambda_{0t}|\mathrel{\stretchto{\mid}{4ex}}\mathcal{F}_t\right)\right)\leq \sum_t\frac{M_X}{\rho}\IE(|\lambda_{1t}-\lambda_{0t}|)=\frac{M_X}{\rho}\IE(\sum_t|\lambda_{1t}-\lambda_{0t}|).$$
	Thus $\IE(\frac{R}{T})$ can again be bounded by sup-norm differences in functions as before and $|\lambda_{10}-\lambda_{00}|$ using~\eqref{eq:KLbdIN}. Next, we need to construct a sieve and construct tests.
	We consider similar sieve 
	\begin{align}
	W_T&=\{K_1,K_2,K_3\alpha,\gamma_1, \gamma_2:K_1\leq K_{1T},K_2\leq K_{2T},K_3\leq K_{3T}, \|\alpha\|_{\infty}\leq A_{T}, \min(\alpha,\gamma_1,\gamma_2)>\rho_T, \nonumber\\&\qquad\max{\gamma_1}+\max{\gamma_2}\leq 1-A_T/B_T, \lambda_0\leq B_T\},
	\end{align}
	as in the previous problem. Within the sieve, we have $ \IE_{\IE_{t-1}}(\max (X_t,\lambda_t)) <B_T.$ Here the extra terms such as $K_{3}$ stands for number of basis in $b_1(t)$ and the vectors $\gamma_1$ and $\gamma_2$ correspond to the B-spline coefficients of the functions $a_1(t)$ and $b_1(t)$ respectively. Also note that we now have a lower bound for $A_T$ for technical need.  We take $\rho_T \approx T^{-a}$ with $a<1$, $A_T=B_T(1-\exp(\log T/T)\rho_T), B_T \approx T^{a_2}$ for sufficiently large $T$ such that $\exp(\log T/T)\rho_T<1$. Within the sieve again we use a variant of above inequality. Note that within the sieve $\IE(X_t)\leq B_T$ and $\IE(\lambda_t)\leq B_T$. 
	
	We have that,
	\begin{align}
	|\lambda_{1t}-\lambda_{t}|&\leq\|\mu_1-\mu\|_{\infty} + X_{t-1}\|a_{11}-a_{1}\|_{\infty}+(1-\frac{A_T}{B_T})|\lambda_{1,t-1}-\lambda_{t-1}| + \lambda_{t-1}|b_{11}-b_{01}|_{\infty}\label{eq:ineq1}
	\end{align}
	
	and also,
	
	\begin{align*}
	\frac{|\lambda_{t}-\lambda_{1t}|}{\lambda_{t}}&\leq\frac{1}{\rho_T}\|\mu-\mu_{1}\|_{\infty} + \frac{1}{\rho_T}\|a_{1}-a_{11}\|_{\infty}+\frac{1-A_T/B_T}{\rho_T}\frac{|\lambda_{t-1}-\lambda_{1,t-1}|}{\lambda_{t-1}} + \frac{1}{\rho_T}\|b_{1}-b_{11}\|_{\infty}
	\end{align*}
	
	By recursion, 
	\begin{align}
	\frac{|\lambda_{t}-\lambda_{1t}|}{\lambda_{t}}\leq\frac{G_T^t-1}{(G_T-1)\rho_T}\left[\|\mu-\mu_{1}\|_{\infty}+\|a_{1}-a_{11}\|_{\infty}+\|b_{1}-b_{11}\|_{\infty}\right] + \frac{G_T^{t-1}}{\rho_T}|\lambda_{0}-\lambda_{01}|,
	\end{align}
	where $G_T=\frac{1-A_T/B_T}{\rho_T}>1$. Since RHS is increasing in $t$ and we only need to find a bound for $t=T$. If $A_T,B_T$ and $\rho_T$ are chosen in such a way that $G_T\asymp\exp(\log T/T)$, then $G_T^T\asymp T$. Based on that $r_1$, $r_2, r_3$ and $r_4$ can be chosen. For sufficiently large $T(>1/a)$ we have $(1-\exp(\log T/T)\rho_T)<1$. Let us assume that $\|\mu-\mu_1\|_{\infty}=r_1, \|a-a_1\|_{\infty}=r_2,\|b-b_1\|_{\infty}=r_3, |\lambda_0-\lambda_{01}|=r_4$. Then for $r_i\leq\frac{\rho_T}{T^{1+a_3}}$, we have that $\frac{|\lambda_{t}-\lambda_{1t}|}{\lambda_{t}}\leq 1/T^{a_3}$ for all $t$ with $a_3>0$. The choice of $a_3$ is shown later. Next goal is to find the radii for which $\IE_{\psi}\left(\frac{f}{f_1}\right)^2$ is bounded. Similar steps as before first give us $\IE_{\psi_1}\left(\frac{f}{f_1}\right)^2\leq \frac{1}{T}\sum_{t=1}\IE_{\psi}\left(\frac{f(X_t|\mathcal{F}_{t-1},\lambda_0)}{f_1(X_t|\mathcal{F}_{t-1},\lambda_0)}\right)^T$ and then the following,
	$$\IE_{\psi} \left(\frac{f(X_t|\mathcal{F}_{t-1},\lambda_0)}{f_1(X_t|\mathcal{F}_{t-1},\lambda_0)}\right)^T\approx\IE_{\psi}\exp\frac{T(\lambda_{1t}-\lambda_{t})(\lambda_{1t}-\lambda_{t})}{\lambda_{t}}\leq \IE_{\psi}\exp(T^{1-a^3}|\lambda_{1t}-\lambda_{t}|)\leq \IE_{\psi}\exp\frac{\lambda_{t}}{T^{2a_3-1}}.$$
	We have by Jensen's inequality, $\IE_{\psi}\exp\left[\frac{\lambda_{t}}{T^{2a_3-1}}\right]\leq \IE_{\psi}\exp\left[\frac{X_{t}}{T^{2a_3-1}}\right]$ as $\lambda_t=\IE_{\psi}(X_t|\mathcal{F}_{t-1},\lambda_0)$.
	We can again show by induction that within the sieve $\IE(e^{qX_t})<e^Q$ for some constant $Q$ following similar argument with $q=T^{1-2a_3}$. We again need $qB_T$ independent of $T$. Hence our choice for $a_3$ will be $a_3=\frac{1+a_2}{2}>1/2$. Thus $q$ is small for sufficiently large $T$ and hence $e^q-1\approx q$. We have from MGF of Poisson,
	\begin{align}
	\IE(e^{qX_t})=\IE(\exp\{\lambda_t(e^q-1)))&\approx\IE(\exp(\mu(t)q+a_1(t)X_{t-1}q)+b_1(t)\lambda_{t-1}q\}\nonumber\\&= \IE(\IE_{t-1}(\exp\{\mu(t)q+a_1(t)X_{t-1}q\})(\exp\{b_1(t)\IE_{t-1}(X_{t-1})q\}))\nonumber\\&\quad\leq\IE(\IE_{t-1}(\exp\{\mu(t)q+a_1(t)X_{t-1}q\})\IE_{t-1}(\exp\{b_1(t)X_{t-1}q\}))\nonumber\\&\qquad\leq\IE(\IE_{t-1}(\exp\{\mu(t)q+a_1(t)X_{t-1}q+b_1(t)X_{t-1}q\}))\nonumber\\&\qquad=\IE(\exp\{\mu(t)q+(a_1(t)+b_1(t))X_{t-1}q\}),\label{recrel}
	\end{align}
	by first Jensen's inequality as $\lambda_t=\IE_{\psi}(X_t|\mathcal{F}_{t-1},\lambda_0)$ and positive correlation between $\exp\{a_1(t)X_{t-1}q\}$ and $\exp\{b_1(t)X_{t-1}q\}$ under the expectation $\IE_{t-1}$. For two positively correlated random variables $Y$ and $Z$ under the sample space, we have $E(YZ) > E(Y)E(Z)$. Now using this recurrence result \eqref{recrel} of $\IE(e^{qX_t})$, we again arrive at similar type of bounds for $r_1\leq\frac{\sqrt{\rho_T}}{\sqrt{T}}, r_2\leq\frac{\sqrt{\rho_T}}{\sqrt{TB_T}}$ to ensure that $\IE(e^{qX_t})<e^Q$ for some constant $Q$ for all $t$. We also need that $r_4\asymp r_1, r_3\asymp r_2$, where $\asymp$ means asymptotically equivalent. Finally we need $r_1\leq\min\{\frac{\sqrt{\rho_T}}{\sqrt{T}},\frac{\rho_T}{T^{1+a_3}}\}$ and $r_2\leq\min\{\frac{\sqrt{\rho_T}}{\sqrt{TB_T}},\frac{\rho_T}{T^{1+a_3}}\}$ and $r_4\asymp r_1, r_3\asymp r_2$. These radii are also of polynomial order in $T$.  Rest of the pieces of the proof follow similar arguments as before.

	\bibliographystyle{bibstyle}
	\bibliography{main}
	
\end{document}